%% file: main.tex
\documentclass{article}

\input{deps.inc}

\newcommand{\name}{\textit{{MAP-Me }}}
\newcommand{\nameu}{\textit{{MAP-Me-IU }}}

\newcommand{\namenosp}{\textit{{MAP-Me}}}
\newcommand{\nameunosp}{\textit{{MAP-Me-IU}}}

\begin{document}

\title{MAP-Me: Managing Anchor-less Producer Mobility in Information-Centric Networks}
\date{}

\author{
  Jordan Augé\thanks{Cisco Systems},
  Giovanna Carofiglio\footnotemark[1],
  Giulio Grassi\thanks{LiP6 / UPMC},\\
  Luca Muscariello\footnotemark[1],
  Giovanni Pau\footnotemark[2]~\thanks{UCLA},
  Xuan Zeng\thanks{IRT SystemX}~\footnotemark[2]
}
\maketitle

\begin{abstract}
	Mobility has become a basic premise of network communications, thereby requiring a native integration into 5G networks. Despite the numerous efforts to propose and to standardize effective mobility management models for IP, the result is a complex, poorly flexible set of mechanisms. 
	
	The natural support for mobility offered by ICN (Information Centric Networking), makes it a good candidate to define a radically new solution relieving limitations of traditional approaches. If consumer mobility is supported in ICN by design, in virtue of its connectionless pull-based communication model, producer mobility is still an open challenge.
	
	In this work, we propose \namenosp, an anchor-less solution to manage micro mobility of content producer via ICN name-based data plane, with support for latency-sensitive applications. First, we analyze \name performance and provide guarantees of correctness and stability.
	Further, we set up a realistic simulation environment in NDNSim 2.1 for \name evaluation and comparison against existing solutions: either random waypoint and trace-driven car mobility patterns are considered under 802.11 radio access. 
	Results are encouraging and highlight the superiority of \name in terms of user performance and of network cost metrics.
\end{abstract}

\section{Introduction}
 
With the phenomenal spread of portable user devices, mobility has become a basic
requirement for almost any communication netwteork as well as a compelling feature to integrate in next generation networks (5G).
The need for a mobility management paradigm to apply within IP networks has striven a lot of efforts in research and standardization bodies (IETF, 3GPP among others), all resulting in a complex access-dependent set of mechanisms implemented via a dedicated control infrastructure. 
The complexity and lack of flexibility of such approaches (e.g. Mobile IP) calls today for a radically new solution dismantling traditional assumptions like tunneling and anchoring of all mobile communications into network core. 

The Information Centric Network (ICN) paradigm brings native support for mobility, security and storage within the network architecture, hence emerging as a promising 5G technology candidate. Specifically on mobility management, ICN  has the potential to relieve  limitations of current approaches by leveraging its primary feature, the redefinition of  packet forwarding based on \emph{names} rather than on \emph{network addresses}. 
We believe that removing the dependence on location identifiers is a first step in the direction of removing the need for any anchoring of communications into fixed network nodes, which may considerably simplify and improve mobility management.

As a direct result of ICN principles, consumer mobility is naturally supported by design: a change in  physical location for the consumer does not translate into a change in the data plane like for IP. The retransmissions of requests for not yet received data by the consumers take place without involving any signaling to the network.

Producer mobility and realtime group communications present more challenges, depending on frequency of movements, on latency requirements and on content lifetime. 

Tackling such problem, in a simple and effective way by exploiting ICN key characteristics, is at the core of this paper. Previous attempts have been made in ICN literature to go beyond traditional IP approaches, by using existing ICN request/data packet structures to trace producer movements and to dynamically build a reverse forwarding path (see \cite{NDN-survey} for a survey). They still rely on a stable home address to inform about producer movements (e.g. \cite{zhang2014kite}) or on buffering of incoming requests at the producer's previous point of attachment (e.g. \cite{ANC_Mobilitysupport}), which prevents support for latency-sensitive applications. 

In this paper, we aim to take one step forward in the definition of a name-based mechanism operating in the forwarding plane and completely removing any anchoring, while striving for latency minimization. 

The main contribution of this work is a proposal for an anchor-less mobility management mechanism, named \namenosp,  with the following characteristics:
\begin{itemize}
\item \name~addresses micro (e.g. intra Autonomous Systems) producer mobility. Addressing macro mobility as well as the interaction with network routing are a \emph{non-goal} of this paper, left for future work.
\item \name~does not rely on global routing updates which would be too slow and too costly, rather works at a faster timescale propagating forwarding updates and  leveraging realtime notifications left as breadcrumbs by the producer to enable live tracking of its position\footnote{For simplicity, we use the word \textit{producer} in place of the more correct expression \textit{producer name prefixes} }.
The main goal being the support for both high-speed mobility and realtime group applications like Periscope~\cite{Periscope}. 
\item \name~leverages core ICN features like stateful forwarding, dynamic and distributed Interest load balancing to update the forwarding state at routers, relaying former and current producer locations.
\item \name~is designed to be access-agnostic, in order to cope with highly heterogeneous wireless access and multi-homed/mobile users.
\item Low overhead in terms of signaling, additional state at routers and computational complexity are also targeted in the design to provide a solution able to scale to large and dynamic mobile networks.
\end{itemize}
 
To evaluate this proposal, we first contribute an analysis of protocol correctness and guarantees, then we provide a realistic simulation environment in NDNSim 2.1~\cite{ndnsim}, where we compare it against ideal Global Routing with instantaneous updates, anchor-based and tracing-based solutions over a set of random waypoint and trace-driven mobility patterns representing V2I scenarios based on 802.11 radio access.
Results show that \name outperforms existing alternatives by either user performance (e.g. loss, delays) and network cost (e.g. signaling overhead, link utilization) metrics.

The rest of the paper is organized as follows: Sec.\ref{sec:relwork} discusses different classes of mobility management approaches for ICN, discussing pros and cons of existing solutions. In Sec.\ref{sec:approach}, we introduce the design principles and detail \name~ operations, before analyzing its correctness and path stretch guarantees in Sec.\ref{sec:analysis}.   
The NDNSim evaluation is gathered in Sec.\ref{sec:evaluation}.
Finally, Sec.\ref{sec:conclusions} concludes the paper.

%

\section{Related Work}
\label{sec:relwork}
In this paper we primarily focus on Named-Data Networking (NDN) architecture~\cite{VANNDN-Conext}, though the proposed mobility management model has broader applicability within ICN designs. In the implemented ICN protocol stack, we consider the jointly optimal multipath Interest control and distributed dynamic forwarding derived in \cite{ICNP, Nomen13}, while node mobility is emulated either via classical mobility models (e.g. Random Way Point) and via trace-driven car patterns~\cite{Navigo}.

\subsection{Mobility management in ICN}

Many efforts have been devoted to define mobility management models for IP networks in the last two decades, resulting in a variety of complex, often not implemented, proposals. A good survey of these approaches is RFC 6301 \cite{RFC6301}. A first survey of ICN mobility management approaches is provided in \cite{NDN-survey}. Authors distinguish three categories of solutions -- routing, mapping and tracing -- based on the type of indirection point (also called Rendez-Vous, RV).
We build on such classification and refine it mainly to distinguish name-based approaches not relying upon the existence of a RV (anchor-less).
We classify mobility management solutions into:
\begin{enumerate}
	\item[\textbf{GR:}] \textit{Global Routing} solutions using the routing plane and requiring to update all network routers about movements of a mobile node;
	\vspace{-2mm}
	\item[\textbf{RB:}] \textit{Resolution-Based or mapping at routing-layer}, solutions involving a resolution of identifiers into locators to be performed at dedicated RV nodes (DNS-like infrastructure);
	\vspace{-2mm}
	\item[\textbf{AB:}] \textit{Anchor-Based or mapping at network-layer} solutions keeping the mapping at network-layer by using a stable home address (or anchor) as a RV. The anchor is kept aware of mobile node movements and is also responsible for tunneling packets to the new location;
	\vspace{-1mm}
	\item[\textbf{TB:}] \textit{Tracing-Based} solutions requiring the mobile node to create a hop-by-hop forwarding reverse path from its RV back to itself by propagating and keeping alive traces stored by all involved routers. Forwarding to the new location is enabled without tunneling; 
	\vspace{-2mm}
	\item[\textbf{AL:}] \textit{Anchor-Less} solutions requiring the mobile node to signal its movements to the network by propagating name-based messages that are kept by all involved routers to guarantee reachability at the new location without requiring any RV.
\end{enumerate}
\vspace{-1mm}
It is a shared concern that GR suffers from scalability issues even in IP. Thus, it may not be considered as a viable solution, in general, in presence of frequent mobility and especially in ICN where the naming space is even larger than in IP. 
In the rest of the paper, we use an ideal global routing approach instantaneously updating all routers' tables as an ideal reference for comparison with other approaches.

 \begin{table*}\label{tab:comparison}
  \centering
 	\begin{tabular}{|c|c|c|c|c|}
 		\hline
 		& Routing & Fwd & RV     & RV     \\
 		& plane   & plane      & signal.& tunnel.\\
 		
 		\hline
 		GR & x &   & x &   \\
 		RB & x &   & x &   \\
 		AB &   & x & x & x \\
 		TB &   & x & x &   \\
 		AL &   & x &  &    \\
 		\hline
 	\end{tabular} \\
 	\begin{tabular}{|c|c|c|c|c|c|c|}
 		\hline
 		& handoff& Path   & n. signal.& Fwd   & Robustness  & Reactivity to\\
 		& delay  & stretch& messages    & state & to failures & mob. freq. \\
 		
 		\hline
 		GR & igh & Low& High& Low& High& Low \\
 		RB & igh & Low& Medium & Low & Medium& Low \\
 		AB & Medium& Medium& Medium& Low& Low& Medium\\
 		TB & Medium& Medium& Medium& igh& Low& Medium\\
 		AL & Low& \textbf{Low}/Medium& \textbf{Low}/Med.& Medium& \textbf{igh}& Low/\textbf{High}\\
 		\hline
 	\end{tabular}
 	\caption{Summary of the comparison of the different classes of mobility protocols.}
 \end{table*}
 
\subsubsection{RB - Resolution-based}

RB proposals rely on the resolution of names into routable location identifiers: a DNS-like mapping system is updated every time the producer moves in the network and guarantees the correctness of the binding between content names and current producer location~\cite{RDV_ermans,RDV_Pursuit,RDV_Whatbenefits,RDV_SCOM,ANC_Mobilitysupport}. Once the resolution performed, packets can be correctly routed along the shortest path, with unitary path stretch (defined as the ratio between the realized path length over the shortest path one). Requiring  explicit  resolution, together with a  strict separation of names and locators, RB solutions involve a scalable ICN  routing infrastructure  able to leverage forwarding hints~\cite{RDV_Hermans, RDV_Whatbenefits}; however, scalability is achieved at the cost of a large handoff delay as evaluated e.g. in~\cite{ANC_Mobilitysupport} due to RV update and name resolution.  To summarize, RB solutions show good scalability properties and low stretch in terms of consumer to producer routing after the name resolution, but result to be unsuitable for frequent mobility and for reactive rerouting of latency-sensitive traffic. 

\subsubsection{AB - Anchor-based}

Anchor-based mobility solutions in ICN are mostly inspired by Mobile IP:  they assume the presence of an anchor the mobile producer is associated to, and sends updates to when he moves. The anchor is assumed to be on the path followed by requests to locate the producer, hence is responsible for intercepting and relaying consumer requests directed to the producer to its current location. 
Examples of anchor-based ICN proposals are~\cite{ANC_MobiCCN, ANC_Device} and one of the approaches compared in \cite{ANC_Mobilitysupport}.

For instance, in MobiCCN~\cite{ANC_MobiCCN}, mobility is handled via geographical routing and via the use of
anchors selected in virtue of the proximity in the hyperbolic space. Name resolution routers are distributed across the network in the form of
anchors, updated in case of mobility and relaying traffic in place of mobile producers. Anchor placement becomes critical 
for the performance of the entire approach. Unlike the above mentioned solution, in~\cite{ANC_Device} a change in the name is required to manage mobility in an anchor-based fashion, with bad consequences in terms of integrity verification, potential degradation in caching efficiency etc.  
All AB solutions have  better reactivity and good path stretch properties w.r.t. RB proposals at the cost of larger signaling overhead, single point of passage weakness (as observed in~\cite{ANC_Mobilitysupport}), preventing ICN multipath and limiting robustness to failure/congestion.

\subsubsection{TB - Tracing-based}

Tracing-based solutions assume data are published under a global stable RV prefix and aim to create a ``breadcrumb trail'' connecting current producer location to its stable RV location in order to   relay consumer Interests to the mobile producer.
Kite~\cite{zhang2014kite} leverages Interest notifications sent by the mobile producer to its RV to build a valid path to its current location via 
traces stored at all traversed routers into the PIT.
Traces have to be maintained over time by keep-alive messages. 
While it exploits ICN data plane features without requiring a separate control infrastructure, Kite involves a large signaling overhead 
due to keep-alive traces that are propagated and stored within a trace table within PIT. 
The idea of creating a reverse path to a stable home router is expressed also in \cite{PMC}, where authors propose a similar tracing-based approach, leveraging updates in FIB, rather than in PIT.

\subsubsection{AL - Anchor-less}

Anchor-less approaches are less common and introduced in ICN to enhance reactivity with respect to AB solutions by leveraging ICN name-based routing.  We can mention in this category the Interest Forwarding approach presented in \cite{ANC_Mobilitysupport} 
and the multicast based proposal in \cite{ALE_Seamless}.
For instance in the \emph{Interest Forwarding} scheme, the mobile producer sends a notification to its current point of attachment (PoA) before moving. 
The PoA starts buffering incoming Interests for the mobile producer until a forwarding update is completed and a new route built to reach the current location of the producer. Enhancement of such solution consider handover prediction.  Besides the potentially improved delay performance w.r.t. other categories of approaches, some drawbacks can be recognized: buffering of Interests may lead to timeouts for latency-sensitive applications and handover prediction is hard to perform in many cases. 

\textit{Remark}: It is worth observing that in-network caching and name-based routing in ICN also enable a routing to replica approach careless of producer movements (referred to as data depot in \cite{NDN-survey}).
Clearly, such approach is not suitable for realtime applications or targeted to unpopular content which may be replaced in cache due to memory limitations. 
A study of the advantages for popular items can be found in \cite{RDV_Whatbenefits}. 
A summary of the characteristics of each class of approaches is reported in Tab.\ref{tab:comparison}.
Based on such qualitative comparison of different families of solutions, we detail in the next section the design principles inspiring our proposal. \name belongs to AL class and aims to overcome limitations of existing alternatives by providing path stretch guarantees and especially targeting low handoff delay and high reactivity to changes (in \textbf{bold} in Tab.\ref{tab:comparison} we highlight the improvements of $\name$ within AL class).

\section{Design} \label{sec:approach}

In this section, we introduce \name, a micro mobility management protocol for
ICN networks. Based on the classification in the previous section, we detail here the
design principles inspiring \name.\\
We recall that we focus on an \textbf{anchor-less name-based layer-2 agnostic approach operating at ICN forwarding plane}.
In the quest for purely AL designs, we target a solution with the following additional characteristics:
\begin{description}
\vspace{-2mm}
\item[Transparent:] no differentiation between mobile consumer/ producer, no
required preliminary knowledge of handover, no name change.
The latter feature is important to avoid issues like triangular routing or
caching degradation.
\vspace{-2mm}
\item[Distributed:] we design \name to be fully distributed, involving routers at the edge of the network and so realizing  effective traffic offload close to end-users. Robustness issues like single point-of-passage problem due to centralized mobility management are prevented and multipath capabilities leveraged for efficient resource utilization. 
\vspace{-2mm}
\item[Lightweight:] we consider prefix granularity in updates, \\  rather than content or chunk granularity, 
in order to minimize signaling overhead and temporary state kept by in-network nodes.
\vspace{-2mm}
\item[Reactive:] we introduce network notifications and discovery mechanisms to support realtime producer tracking to support latency-sensitive communications.
\vspace{-2mm}
\item[Robust:] to network conditions (e.g. wireless or congestion losses and delays), by leveraging hop-by-hop retransmissions.
\end{description}

\subsection{\name description}

\name definition assumes the existence of a routing protocol responsible for creating/updating
the Forwarding Information Base (FIB) of all routers, possibly with multipath routes, and for managing network failures. \name~ builds on top of such stable routing configuration to dynamically handle producer mobility events via temporary FIB updates with the objective of minimizing producer unreachability. The rationale behind \name is to let the producer announce its movements to the network by sending a message to ``itself'', namely a message carrying its name prefix(es) that the network forwards according to the latest information stored in the FIBs of traversed routers.   Moreover, \name makes use of route versioning by carrying sequence numbers in mobility update messages. This allows to properly manage multiple route updates by taking into account their freshness.
Sequence numbers associated to a given route  can be stored as a regular field in the FIB or in additional entries, 
in a memory partition that we call \textit{Temporary FIB} (TFIB). Without loss of generality, we do not specify in this paper if and how
the routing protocol might interpret and exploit route versioning. We also assume that a producer makes use of a single
face to serve consumers.

The key idea of \name~is to avoid relying on a stable home address (as opposed to TB approaches for instance) and rather use name-based forwarding state created by ICN routing protocols or left by previous mobility updates, to switch FIBs on-the-fly to point to the correct new location of the producer. Let us describe the basic functioning of \namenosp, i.e. its mobility update protocol.

\subsection{Update protocol}

To illustrate \name update protocol, we consider an example with a single producer serving prefix \verb+/p+ and moving from position $P_0$ to position $P_1$. Network FIBs are assumed to be populated by a name-based routing protocol. \\
Without loss of generality we assume that routes' versioning  is set to the sequence number zero if the routing protocol
is capable of managing route versioning, otherwise a new temporary FIB entry would be used to manage \name~update messages.
Once the physical handover is completed, the producer sends an Interest Update (IU) with a positive sequence number (1 in this case) towards \verb+/p+ , i.e. towards its previous location according to routing-based FIB state. The sequence number is kept and updated at producer side. While propagating the IU, FIBs at intermediate routers, still pointing towards the previous location $P_0$, are updated with IU incoming face and the IU again forwarded in a hop-by-hop fashion.
IU process stops when the IU reaches $P_0$ and there is no next hop where to forward it.
Figure~\ref{fig:proposalIU} shows the state of the network during the update phase, after the IU has updated some routers and while
being propagated towards $P_0$. 
Basically, the update process consists in re-rooting the forwarding tree originally rooted in $P_0$ in the new producer location, $P_1$ in such example. We will analyse the properties of such process in Sec.\ref{sec:analysis}. 

\subsubsection{Concurrent updates and IU loss management}

It happens in practice that IU-traversed routers have different sequence numbers associated to IU prefix as a result of partial and concurrent updates; also IUs may be lost, e.g. on wireless channels.
We describe here the set of procedures related to the management of concurrent updates and of IU losses.

\textit{Concurrent updates}: To prevent inconsistencies due to concurrent updates (e.g. caused by frequent mobility) \name leverages the sequence number kept by the producer and added to subsequent IU messages. 
Modification of FIB entries is triggered only when the received IU carries a higher sequence number, while the information of an older IU is simply discarded and the same IU sent back to the originating interface with the updated sequence number. 
In case of equal sequence number (which may happen in case of multipath), the new face is added in FIB, without further propagating the IU (Interest aggregation). 

\textit{IU loss management and TFIB}: An IU loss event may break the update process and cause forwarding inconsistencies. Therefore, we introduce a hop-by-hop IU retransmission mechanism leveraging hop-by-hop IU ACKs and soft state kept after IU forwarding in what we denote as Temporary FIB (TFIB). 

Upon each IU forwarding, a couple $(face,timer)$ is added to the TFIB entry corresponding to the given prefix, with timer equal to a retransmission timer that can be adapted based on IU/ACK observed round trip time. 
If the timer expires before reception of the ACK, a new IU transmission is scheduled, while the added information is removed at ACK reception. It is worth noticing that such IU retransmission mechanism is not dependent of update completion, but can be performed in a purely distributed fashion at every router. 

\begin{figure}[htbp]
	\centering\includegraphics[width=0.85\textwidth]{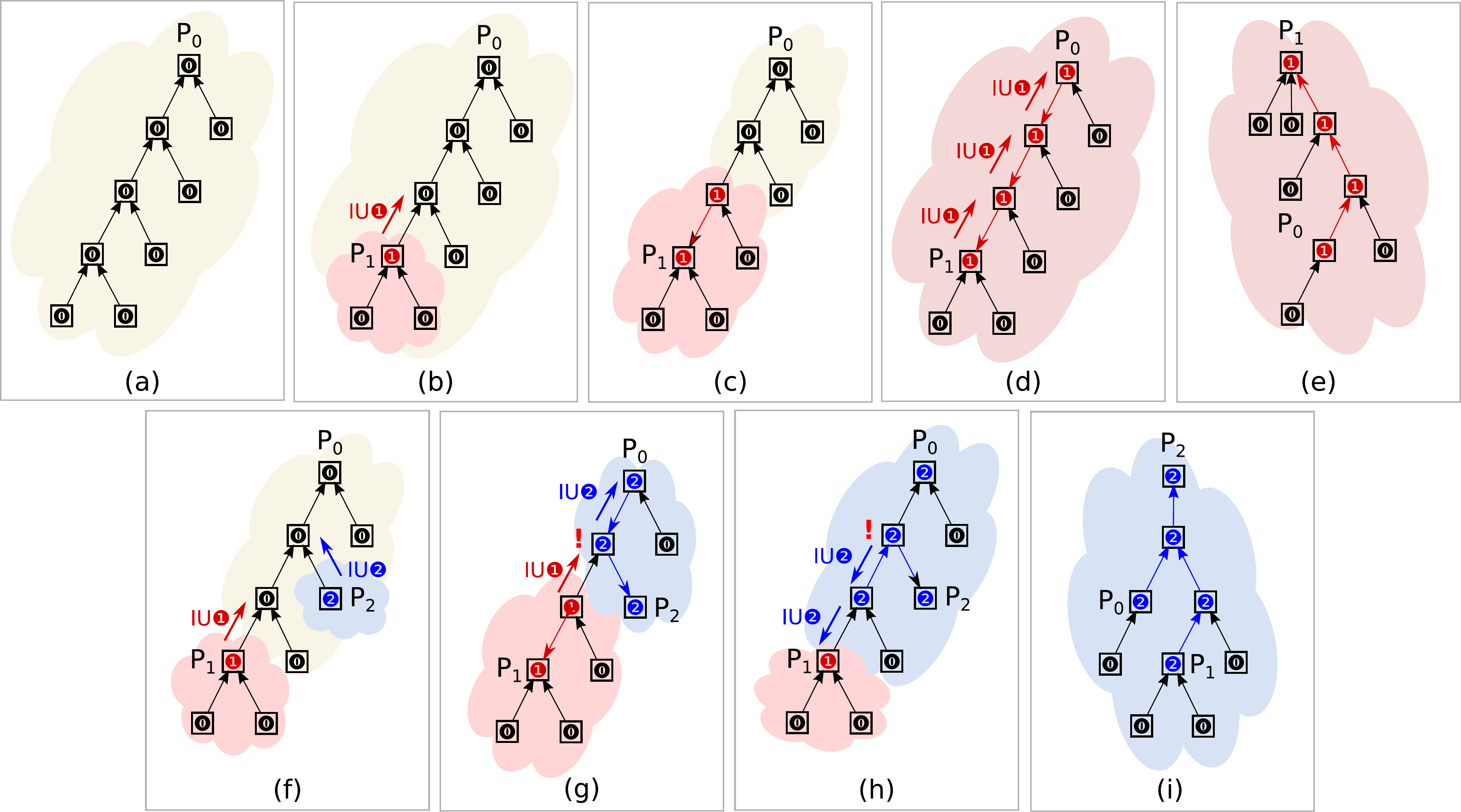}
	\caption{Interest Update propagation example.}
	\label{fig:proposalIU}
\end{figure}
\begin{figure}[htbp]
	\centering\includegraphics[width=0.5\columnwidth]{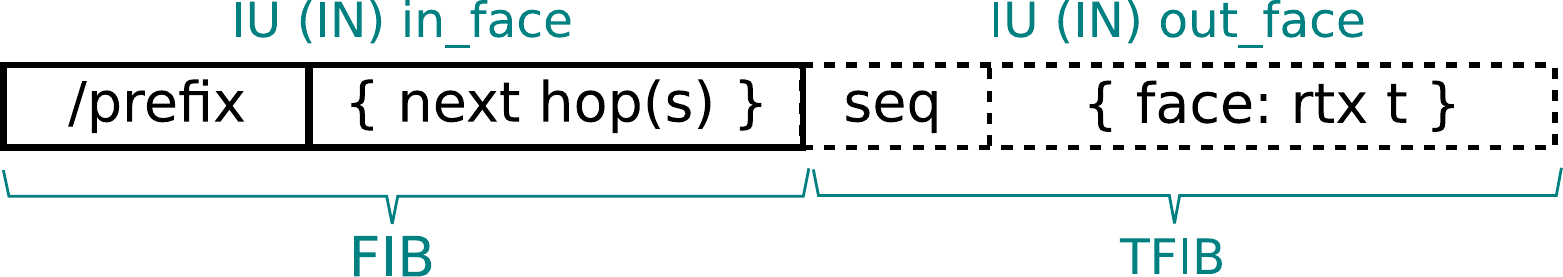}
	\caption{\name FIB/TFIB description.}
		\label{fig:proposalIU-data}
\end{figure}
Algorithm~\ref{alg:update} presents the complete update forwarding process, which is resilient to all such events. Forwarding of regular (consumer) Interests is unchanged.

\begin{algorithm}[H]
	\small
	\caption{ForwardUpdate(UpdateInterest $U$, Ingress face $F$)}\label{alg:update}	
	CheckValidity() \par
	
	\Comment{Get FIB and TFIB entries associated to the prefix} \par
	\Let{$\upepsilon, T$}{FIB.LongestPrefixMatch($U$.name)} \par
	\uIf{ $\upepsilon$ = \o}{
		\Let{$T$}{new TFIBEntry()} \par
		\Let{$\upepsilon$}{FIBEntry.New(Nextop=[$F$], Seq=0, Data=$T$)} \par
		Fib.Add($\upepsilon$)
	}
	
	\If{$U$.seq $\ge \upepsilon$.seq} {
		\Comment{Send ack back}
	}
	
	\uIf{$U$.seq > $\upepsilon$.seq} {
		\Let{$\upepsilon$.seq}{$U$.seq} \par
		
		\Comment{Forward special interest reliably by setting ReTx timer} \par
		\ForAll{ $f \in T$.PendingUpdates} {
			$f$.Send($U$) \par
			$T$.PendingUpdates[$f$] = new RetxTimer($\tau$)
		}
		
		\Let{isNextop}{False} \par
		\ForAll{$f \in \upepsilon$.Nextops} {
			\uIf{$f = F$} {
				\Let{isNextop}{True} \par
				\emph{continue} \par
			}
			$f$.Send($U$) \par
			$T$.PendingUpdates[$f$] = new RetxTimer($\tau$)
		}
		
		\uIf{ ! isNextop} {
			\uIf{$F \in T$.PendingUpdates} {
				ResetTimer($T$.PendingUpdates[$F$]) \par
			}
			$\upepsilon$.Nextops = [$F$]
		}
	}
	\uElseIf{$U$.seq = $\upepsilon$.seq} {
		$\upepsilon$.Nextops.Add($F$)
	}
	\uElse {
		\Comment{Send corrected updated backwards reliably} \par
		$U$.seq = $\upepsilon$.seq \par
		$F$.Send($U$) \par
		$T$.PendingUpdates[$F$] = new RetxTimer($\tau$)
	}
\end{algorithm}
If a producer is temporarily unavailable (e.g. out of range) or disconnects from a node (e.g. handover) the corresponding FIB entries are not removed. 
\name is allowed to update existing FIB entries, but not to remove them. It is up to the control and management plane to take care of longer term FIB updates, whereas \name manages network connectivity between consecutive routing updates.

\subsection{Map-Me Enhancements:\\ Notifications and Discovery}

IU propagation in the data plane helps considerably to accelerate forwarding state re-convergence w.r.t. global routing (GR) or resolution-based (RB) approaches operating at routing plane and to anchor-based (AB) approaches requiring traffic tunneling through the anchor. Still, network latency makes IU completion not instantaneous and before an update completes, it may happen that a portion of the traffic is forwarded to the previous point-of-attachment and dropped because of the absence of valid output face leading to the producer. To prevent such losses, previous work in the AL category have suggested the buffering of Interests at previous producer location (\cite{ANC_Mobilitysupport}). owever, such solution is not suitable for applications with stringent latency requirements (e.g. realtime) and may be incompatible with IU completion times. Moreover, the negative effects on latency performance might
be further exacerbated by IU losses and consequent retransmissions, due to the wireless medium.
To alleviate such issues, we introduce two separate enhancements to \name~update protocol, namely 
\textit{(i)} a \textit{notification} mechanism for frequent, yet lightweight, signaling of producer movements to the network and
\textit{(ii)} a scoped \textit{discovery} mechanism for consumer requests to proactively search for producer recently visited locations.\\

An \textbf{Interest Notification} (IN) is a breadcrumb left by producers in every encountered point-of-attachment. It differs from IU only for the fact that it is not propagated further than the point-of-attachment and it carries a special flag in order to be identified. They shares the same sequence number with IUs (producers update it at every sent message, IU or IN) and follows same FIB lookup and update process. In contrast, however, the trace left by INs in TFIB is the list of next hop faces present in FIB before IN reception and has a timer equal to $0$, which is required to distinguish the nature of face information (IN or IU) in TFIB during the discovery phase described below.

It is worth observing that updates and notifications serve the same purpose of
informing the network of a producer movement. owever, IU process has higher signaling cost than IN's due to message propagation. 
The presence of both mechanisms allows to control the trade-off between protocol reactivity and stability of forwarding re-convergence. 

\textbf{Discovery}.
\name relies on low-latency dedicated direct links among neighboring point-of-attachments typical of mobile deployments\footnote{This is for instance the case of X2 links in LTE deployments. Such deployment is also common in multi-access WiFi networks where access points are logically under the same controller.} to enable a local discovery phase: when a consumer Interest reaches a point-of-attachment with no valid output face in the corresponding entry, the Interest is tagged with a special `discovery' flag, labeled with the latest sequence number stored in FIB (to avoid loops) and broadcasted at one hop to all neighbors along the dedicated links (i.e. LTE X2 links), as described in Algorithm~\ref{alg:forward}.

Once received, if there is a match with older or equal sequence number information in FIB, the `discovery' Interest is discarded. Otherwise, it is either forwarded on available output faces to eventually reach the producer or broadcasted again at one hop along the dedicated links in presence of $0$ timers in TFIB. 
The latter is the case of consumer Interest intercepted by a point-of-attachment updated by a IN rather than IU (hence with no available associated output face). 
 
Changes to Algorithm~\ref{alg:update} are minor since they consist in not forwarding the
IN message, and in setting to $0$ the retransmission timers in the TFIB\footnote{TFIB entries are kept during an update process in progress to correctly route further updates and preserve
network connectivity}. 
As further shown in Sec.\ref{sec:evaluation}, such a mechanism is important
to preserve the performance of flows in progress, especially when latency-sensitive. 
In the rest of the paper, we will evaluate a combined update/notification and discovery approach consisting in
sending an IN immediately after an attachment and an update at most every $T_U$ seconds, referred to as \namenosp, in order to
reduce signaling overhead especially in case of high mobility. The update-only
proposal, denoted as \nameunosp, will also be evaluated to distinguish the gains due to different protocol components.
\begin{algorithm}[H]
	\small
	\caption{InterestForward(Interest $I$, Origin face $F$)
		\label{alg:forward}}
	
	\Comment{Regular CS and PIT lookup} \par
	
	\Let{$\upepsilon$}{FIB.LongestPrefixMatch($I$.name)} \par
	\uIf{ $\upepsilon$ = \o}{
		return \par
	}
	
	\Comment{Discovery interests: discard if no progression} \par
	\If{$I$.seq != \o}{
		\If{$I$.seq $\ge \upepsilon$.seq}{
			return \par
		}
		\Let{$I$.seq}{\o} \par
	}
	
	\uIf{hasValidFace($\upepsilon$.Nextops)}{ \par
		ForwardingStrategy.process($I$, $\upepsilon$)
	}
	\uElse{
		\Let{$I$.seq}{$\upepsilon$.seq} \par 
		SendToNeighbors($I$) \par
	}		
\end{algorithm}

\subsection{Map-Me state and messages}

For the update protocol, we introduce an Interest Update (IU) message, namely an Interest packet carrying producer
prefix (one IU per non-aggregable prefix).
A flag to indicate the request for a special treatment of the packet and a sequence number field are added, the latter to prevent loops and to handle concurrent updates during IU propagation. 
Interest Notification (IN) messages only differ in terms of the special flag that prevents re-forwarding.
FIB entries are simply enriched with a sequence number, set to $0$ by routing protocol and updated according to the number carried by IU/IN messages.
Temporary data about not-yet-acknowledged updates/ notification is stored in what we denote as \textit{Temporary FIB}, TFIB, in order to ensure the reliability and robustness
of the update/notification process.
Each TFIB entry is composed of a list of $(face, timer)$ couples corresponding to the faces where IU have been sent but not yet acknowledged and associated retransmission timer. The timer is set to $0$ for faces corresponding to INs.
FIB/TFIB structure is sketched in Fig.\ref{fig:proposalIU}.

\subsection{Implementation considerations}
A mobility protocol requires some in-network state and update messages to maintain such state.
\name requires a name prefix entry in the FIB/TFIB maintaining the usual forwarding information
plus versioning. During the update process the entry must also store timing information from neighbors
as described in the previous section. The update mechanism (beyond name lookup and forwarding) is a constant 
delay operation.
It is important to note  that in order to prevent malicious updates the first node receiving the
IU from a mobile producer must check its signature to verify the ownership of the name prefix.

\section{Analysis}
\label{sec:analysis}
In this section, we investigate \name~guarantees of forwarding 
update correctness and path stretch stability and
support them by numerical evaluation over known ISP network topologies.
For sake of clarity the analysis reports the formal proofs  only
 in case of single path routing.\\
The extension to multipath routing is straightforward by using
 directed acyclic graphs instead of trees.

\subsection{Correctness and stability of IU mechanism}\label{sec:guarantees}
We consider $m$ consecutive movements of the producer in network positions 
$\{P_0,P_1,...,P_m\}$ and focus on forwarding state variations determined by \name~at the time 
instants corresponding to either producer movements or Interest Update processing. We observe 
that at any of such instants, as illustrated in Fig.\ref{fig:proposalIU}, the network is 
partitioned into a set of islands, whose number varies in $\left[ 1,m+1 \right]$ as a function 
of producer movements and, hence, of the number of ongoing update processes. We assume that 
global routing guarantees the existence of a spanning tree (SP) rooted in the original location 
$P_0$ at the beginning of the mobility process. The tree is not required to be a minimum SP 
or a shortest-path tree. About the completion of the update process after a given movement $k$, 
we can state that
\begin{proposition}
\name update mechanism guarantees finite completion time of update $k$, $\forall k \in \left[1,m\right]$ in a bounded number of hops equal to $2\left( \max_{0\le j<k} (|P_k-P_j|-1)\right)$;
\end{proposition}
\begin{proof}
Assuming IU losses are handled by the retransmission mechanism described in Sec.\ref{sec:approach}, the hop-by-hop propagation of an IU has two possible outcomes: either (i) the next router has a sequence number which is inferior to the IU carried sequence number and in this case IU continues its propagation towards the root of the latest routed tree, 
or (ii) the router has a more recent sequence number, hence the IU is sent back with the 
encountered higher sequence number towards the originating routed position of the producer. Since the maximum sequence number is bounded by $m$, the maximum number of hops traversed by IU with sequence number $k$ is finite.
More precisely, the maximum number of hops traversed by IU with sequence number $k$, $IU_k$ is bounded by twice the maximum distance between the originating router $P_k$ and the farthest previous location $P_j$, $j<k$ minus one, i.e. $2\left( \max_{0\le j<k} (|P_k-P_j|-1)\right)$. 
Indeed, the worst case occurs when $IU_k$ encounters a more recent update $k'>k$ at the hop before reaching the latest routed previous location, which can also coincide with the farthest one in terms of distance. In such case, $IU_k$ propagates back to $P_k$ carrying $k'$ sequence number before stopping. 
\end{proof}
After $IU_k$ propagation, the router $P_k$ and all its predecessors traversed by $IU_k$ to reach the last routed location are connected to the island of highest encountered sequence number, and thus the number of distinct islands reduced by one unit. By iterating the same process on all IUs, it is straightforward to see that at $IU_m$ completion $m+1$ islands associated to sequence number $0,1, ..., m-1$ will have merged into the island created by $IU_m$. About the properties of such island, we can state the following.
\begin{proposition}\label{proposition:tree}
Given a sequence of $m$ consecutive movements of producer position on the routing tree rooted in $P_0$, producer movement $m$ induces a new tree rooted in $P_m$.
\end{proposition} 
\begin{proof}
The initial spanning tree is a rooted directed tree in $P_0$ giving the routes to reach 
all nodes in the network. \name~update mechanism after movement $m$ flips all directed 
links from $P_m$ up to the latest routed position $P_j$, $j<m$, so that they point to 
$P_m$. In presence of multiple concurrent updates, the most recent one, i.e. the one with 
highest sequence number also propagates back along the routes of encountered previous updates. Thus, update completion results in fully merging different rooted trees into the one of highest sequence number, $m$, rooted in $P_m$. 
\end{proof}
\begin{corollary}
\name~is loop free under loop free global routing.
\end{corollary}
\begin{proof}
Starting from the spanning tree given by global routing, Prop.\ref{proposition:tree} states that \name~induces a new tree, as it only flips all edges over the unique path from the original position to the new one.
Indeed, given the unchanged number of links/nodes, the result is still a directed tree rooted in the new position. ence, it is loop free.
\end{proof}
\begin{proposition}
\name~path stretch for node $i$ over the tree rooted in $P_m$, created after producer $m$-th movement, is upper bounded by the ratio $(|i-P_0|_{P_0}+|P_0-P_m|_{P_0})/|i-P_m|_{P_m}$ as $m \rightarrow \infty$, which corresponds to the path stretch of the anchor-based approach with anchor in $P_0$.
\end{proposition}
\begin{proof}
We can distinguish two cases according to whether $P_0$ is on the path between $i$ and $P_m$ on the $P_m$-rooted tree or not. 
If it is, then the path between $i$ and $P_m$ may be split into the paths $i$ to $P_0$ and $P_0$ to $P_m$.
The second component is equal to the path length between $P_m$ and $P_0$ on the initial tree (only directions have been flipped).

The first one 
corresponds to the same path on the initial tree even in terms of directions.
Therefore, the path stretch in this case is exactly equal to $(|i-P_0|_{P_0}+|P_0-P_m|_{P_0})/|i-P_m|_{P_m}$.
Otherwise, if $P_0$ is not on the path between $i$ and $P_m$, the path between $i$ and $P_m$ is, by definition of \name~update process (that utilizes shortest path routing for IUs), shorter than the one including the detour via $P_0$ on the initial $P_0$-rooted tree.
The bound remains true as $m \rightarrow \infty$, because it is intrinsically related to the properties of the initial tree.
\end{proof}

\begin{figure}[htbp]
	\subfigure[AS 1221 stretch evolution]{
		\includegraphics[width=0.45\textwidth]{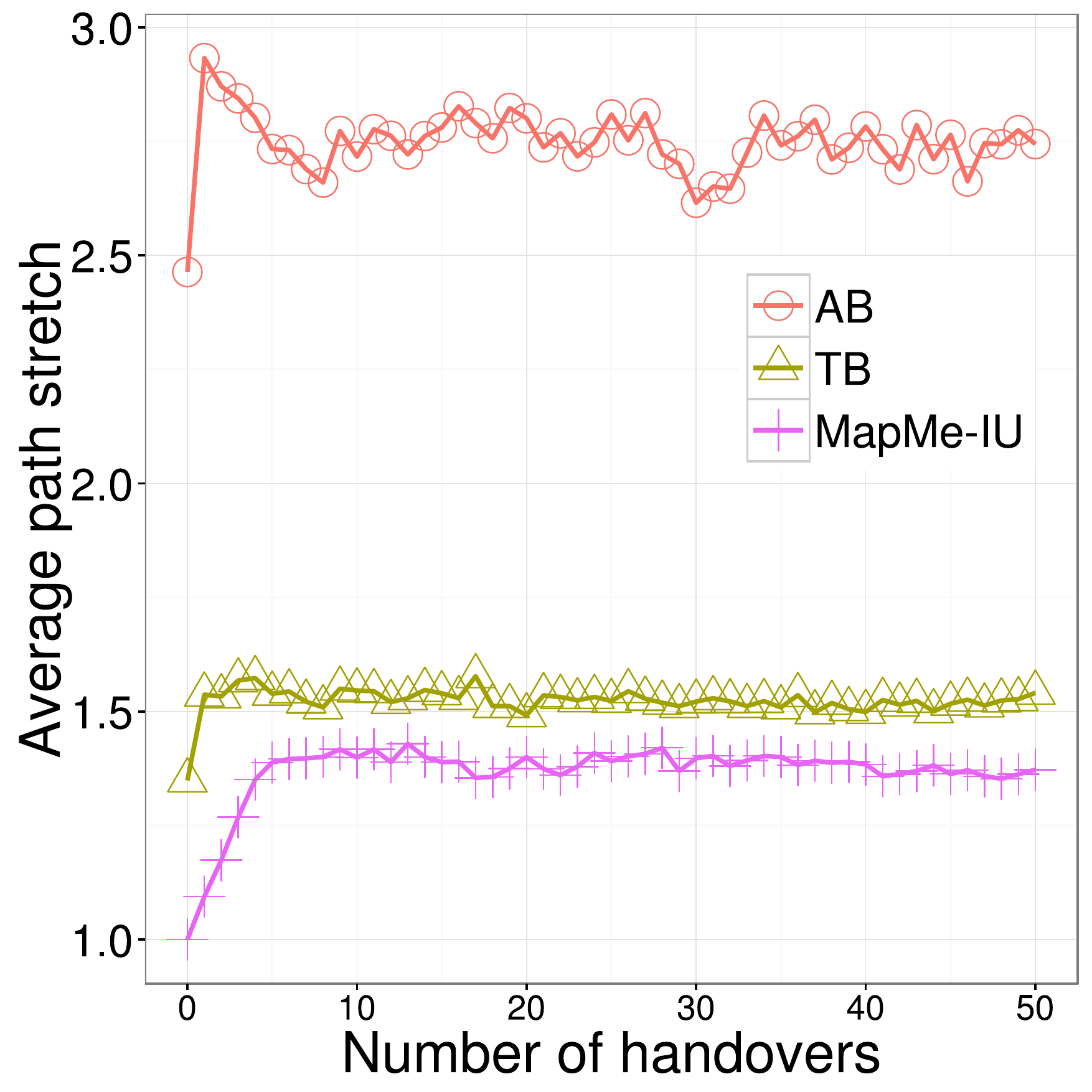}
		\label{fig:stretch-evolution-1221}
	}
  \hfill
	\subfigure[Stretch comparison]{
		\includegraphics[width=0.45\textwidth]{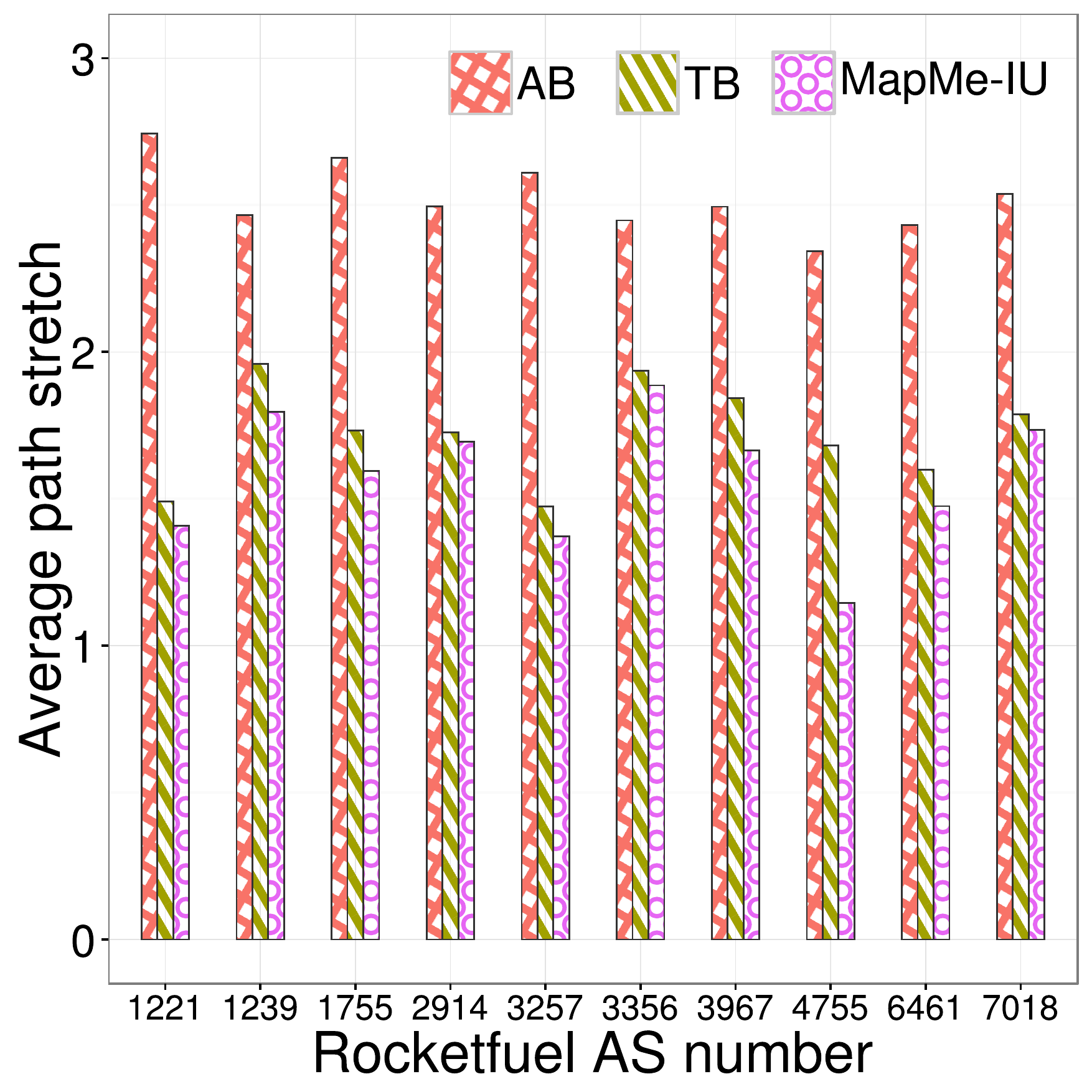}
		\label{fig:stretch-waypoint}
	}
	\caption{Path stretch evolution (a) and comparison (b) over Rocketfuel topologies.}
\end{figure}

\subsection{Numerical Evaluation of path stretch}
We compute now the average path stretch obtained by AB, TB
and \name without notifications (the latter are evaluated by simulation in Sec.\ref{sec:evaluation}), on the topologies obtained from the Rocketfuel project\footnote{\url{http://research.cs.washington.edu/networking/rocketfuel/}. 
We extract the undirected graph corresponding to the largest connected
component as in \cite{ANC_MobiCCN, zhang2014kite}. The differences in resulting datasets
prevent a comparison of results.}. 
The initial position on the consumer, producer and eventual anchor are chosen randomly, and
two mobility models are implemented : (i) uniform, where the producer can jump
towards any other node from the graph, and (ii) random waypoint (RWP), where the producer
chooses a waypoint like in the previous approach, but advances hop-by-hop on the shortest path towards that waypoint, then starts again. 
We average over  $1000$ runs in order to compute ensemble average for the path stretch after $k$ movements of the producer with small confidence intervals. 
Figure~\ref{fig:stretch-evolution-1221} represents the evolution of average
\name path stretch over AS 1221 topology under RWP (other patterns show similar
trends). We observe that path stretch stabilizes beyond $10$ movements, due to the fact that 
\name~preserves the initial structure of the forwarding tree (it only modifies
links direction).  Other Rocketfuel topologies show quantitatively the same
results.  A comparison of the different approaches over the all 10 Rocketfuel
topologies is in Fig.\ref{fig:stretch-waypoint}. 
Under both uniform and RWP mobility, \name outperforms AB, achieving up to
$55\%$ stretch reduction, as well as TB.

\section{Evaluation} \label{sec:evaluation}
\subsection{Simulation setup}

The section gathers simulation assessment of \name~performance over different network scenarios and mobility patterns. 
To this aim we implemented \name~(full approach and with updates only, denoted as MAP-Me-IU), anchor-based (AB),  an example of tracing-based (TB) based on Kite (\cite{zhang2014kite}), and global routing (GR) approaches in NFD within the NDNSim 2.1 framework. We leave out of the comparison resolution-based (RB) or existing AL-solutions leveraging buffering, as they fail to support latency-sensitive applications (see Tab.\ref{tab:comparison} in Sec.\ref{sec:relwork}).
\begin{figure}[htbp]
	\centering
	\includegraphics[width=0.6\columnwidth]{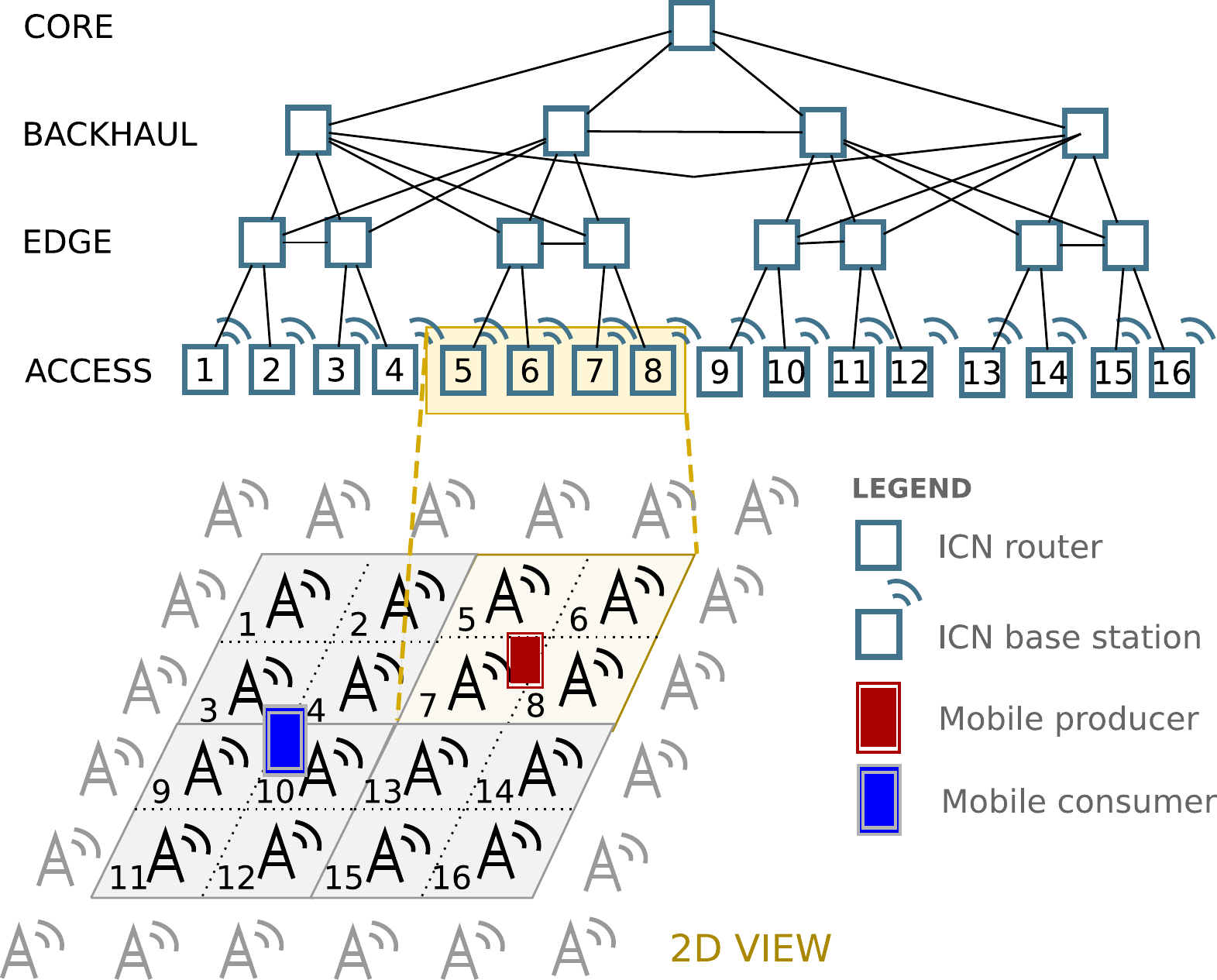}
	\caption{Network with link capacity C=10Mb/s.}\label{fig:topology}
  \vspace{-3mm}
\end{figure}

\begin{figure*}[htbp]
	\centering
	\subfigure[]{
		\includegraphics[width=0.45\textwidth]{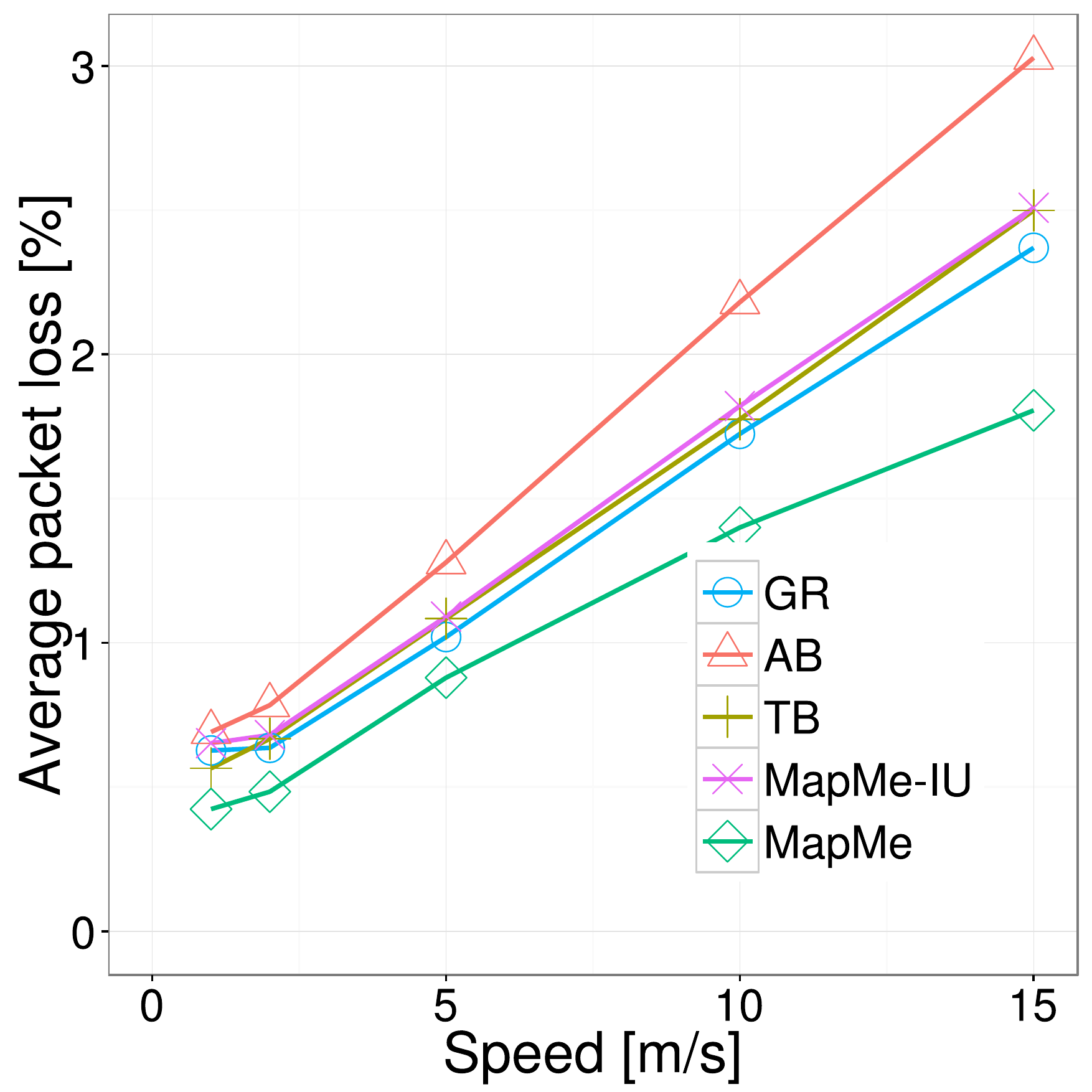}
		\label{fig:loss_rate}
	}
	\hfill
	\subfigure[]{
		\includegraphics[width=0.45\textwidth]{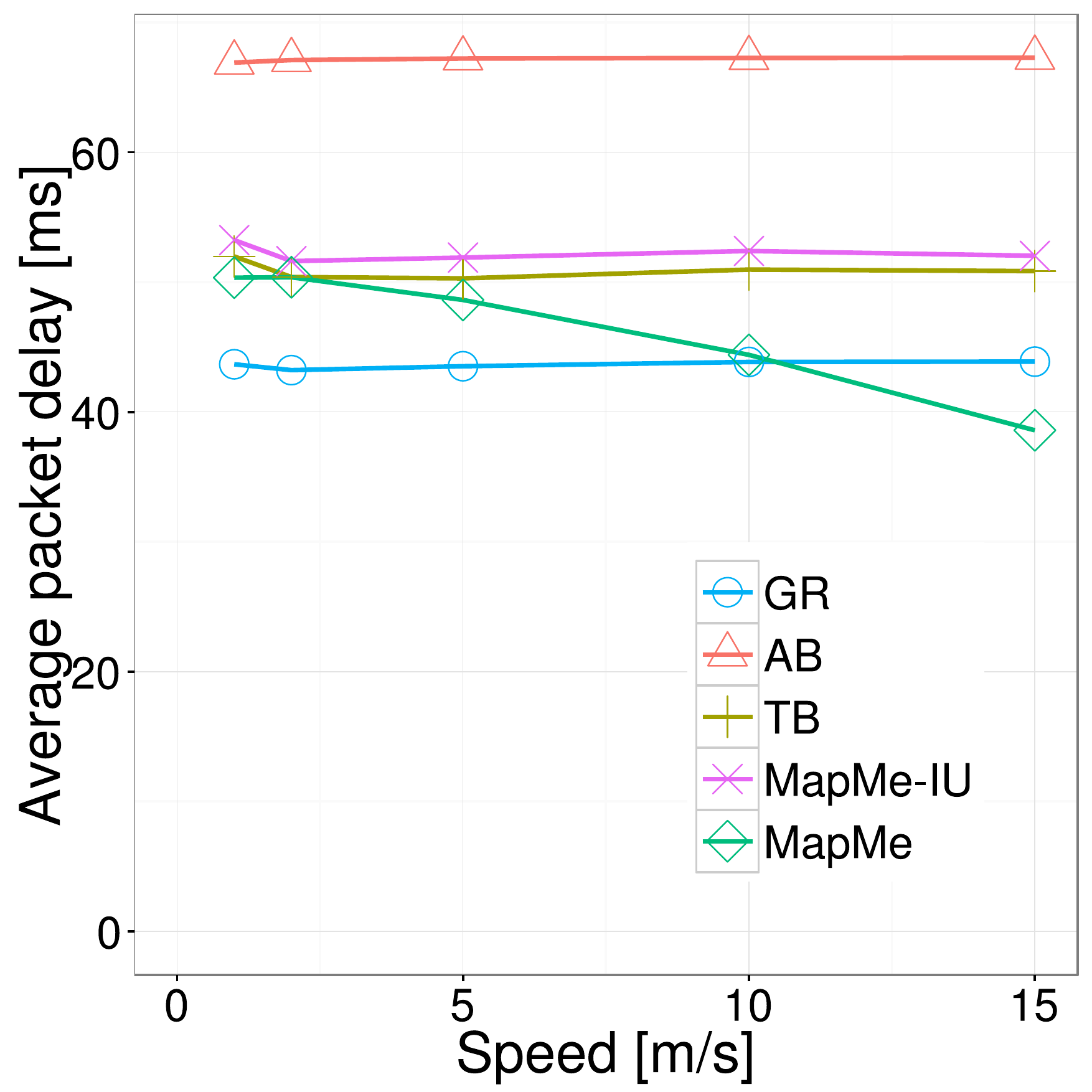}		
		\label{fig:delay}
	}
	\hfill	
	\subfigure[]{		
		\includegraphics[width=0.45\textwidth]{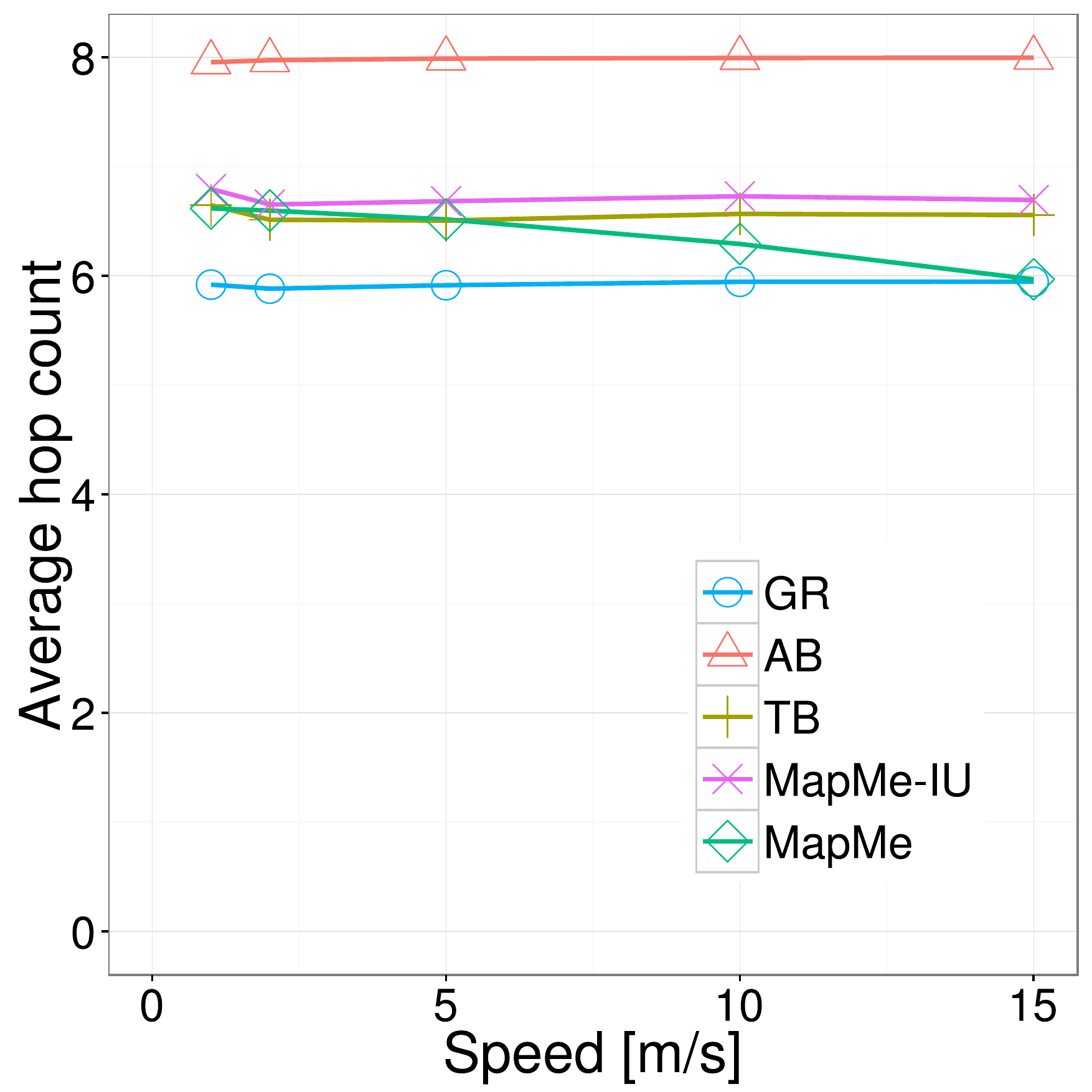}
		\label{fig:stretch}
	}
	\hfill
	\subfigure[]{
		\includegraphics[width=0.45\textwidth]{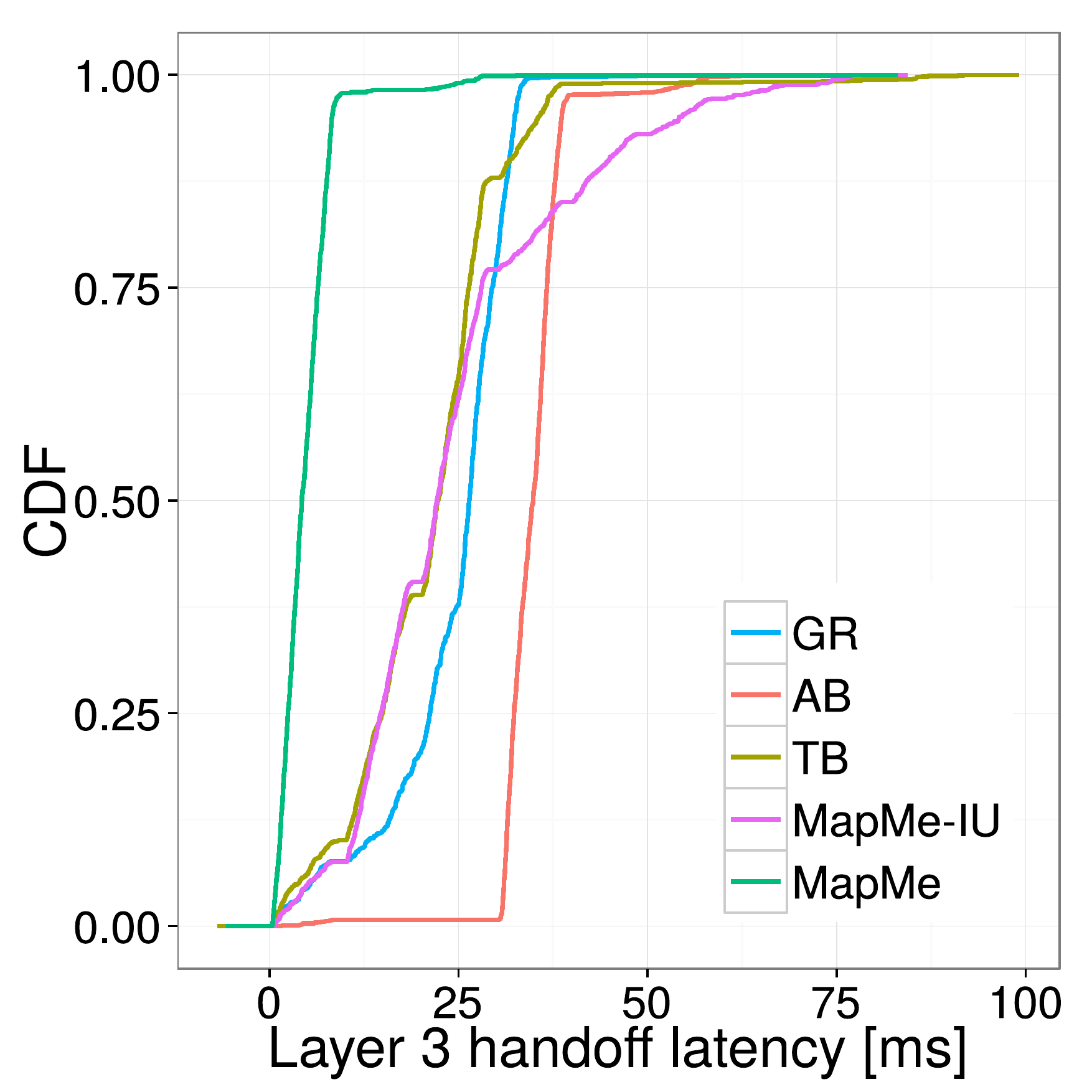}
		\label{fig:l3delay}	
	}
	\vspace{-2mm}
	\caption{User performance. packet loss (a), delay (b) and hop count (c); CDF L3 handoff latency (d).}
	\label{fig:user-network}
\end{figure*}

\begin{figure*}[htbp]
	\centering	
	\subfigure[] {
		\includegraphics[width=0.45\textwidth]{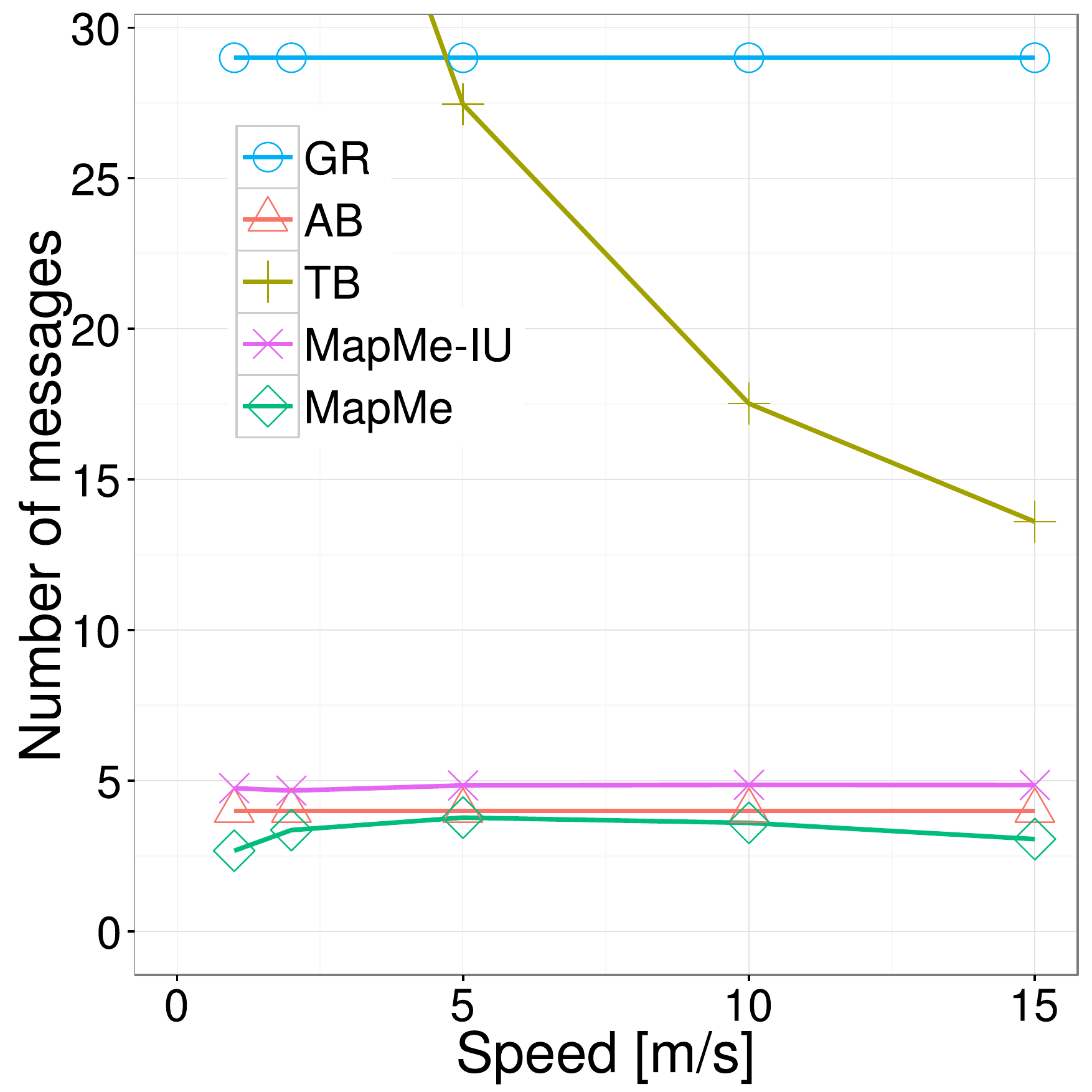}
		\label{fig:overhead}
	}
	\hfill
	\subfigure[] {
		\includegraphics[width=0.45\textwidth]{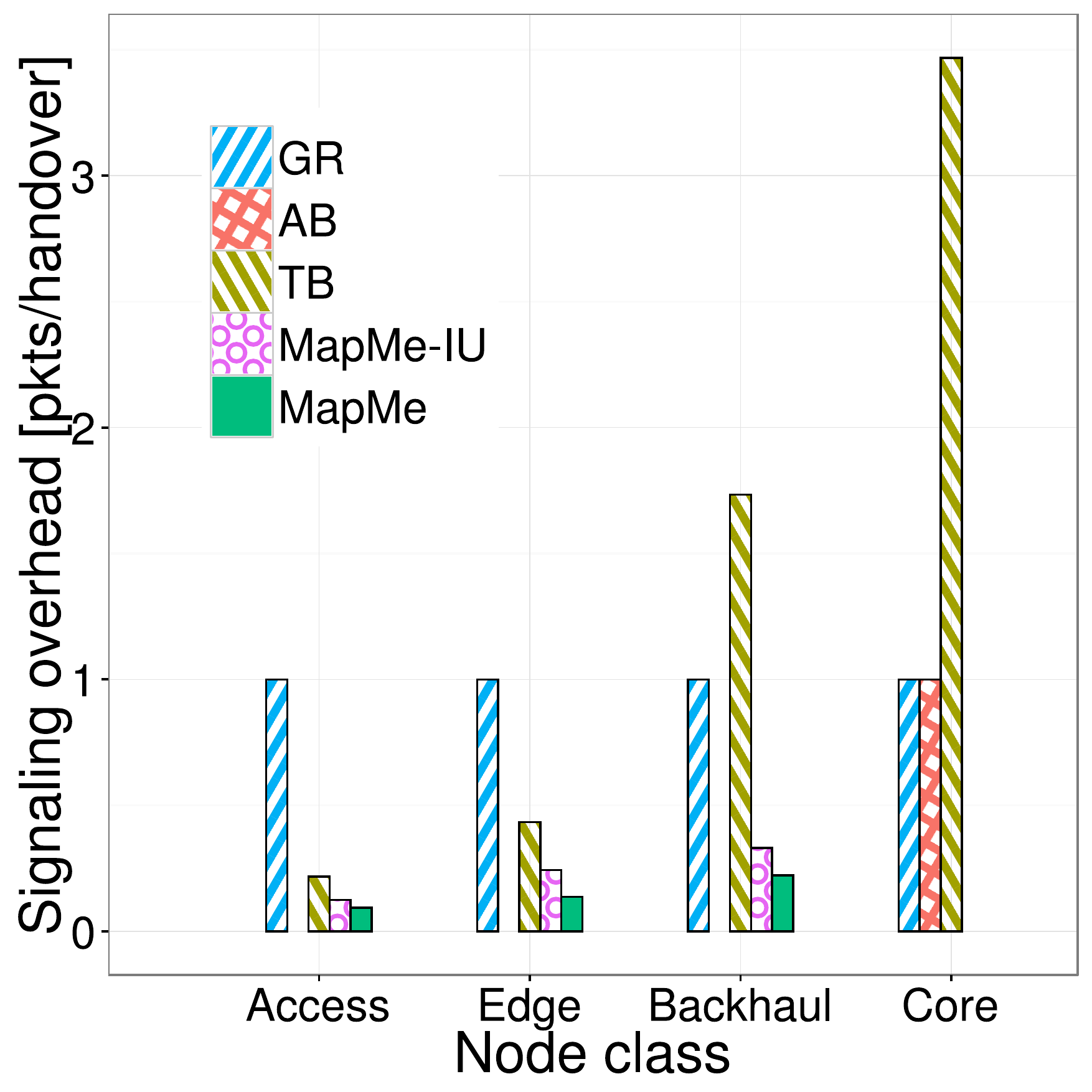}	
		\label{fig:ratio}
	}
	\hfill
	\subfigure[] {
		\includegraphics[width=0.45\textwidth]{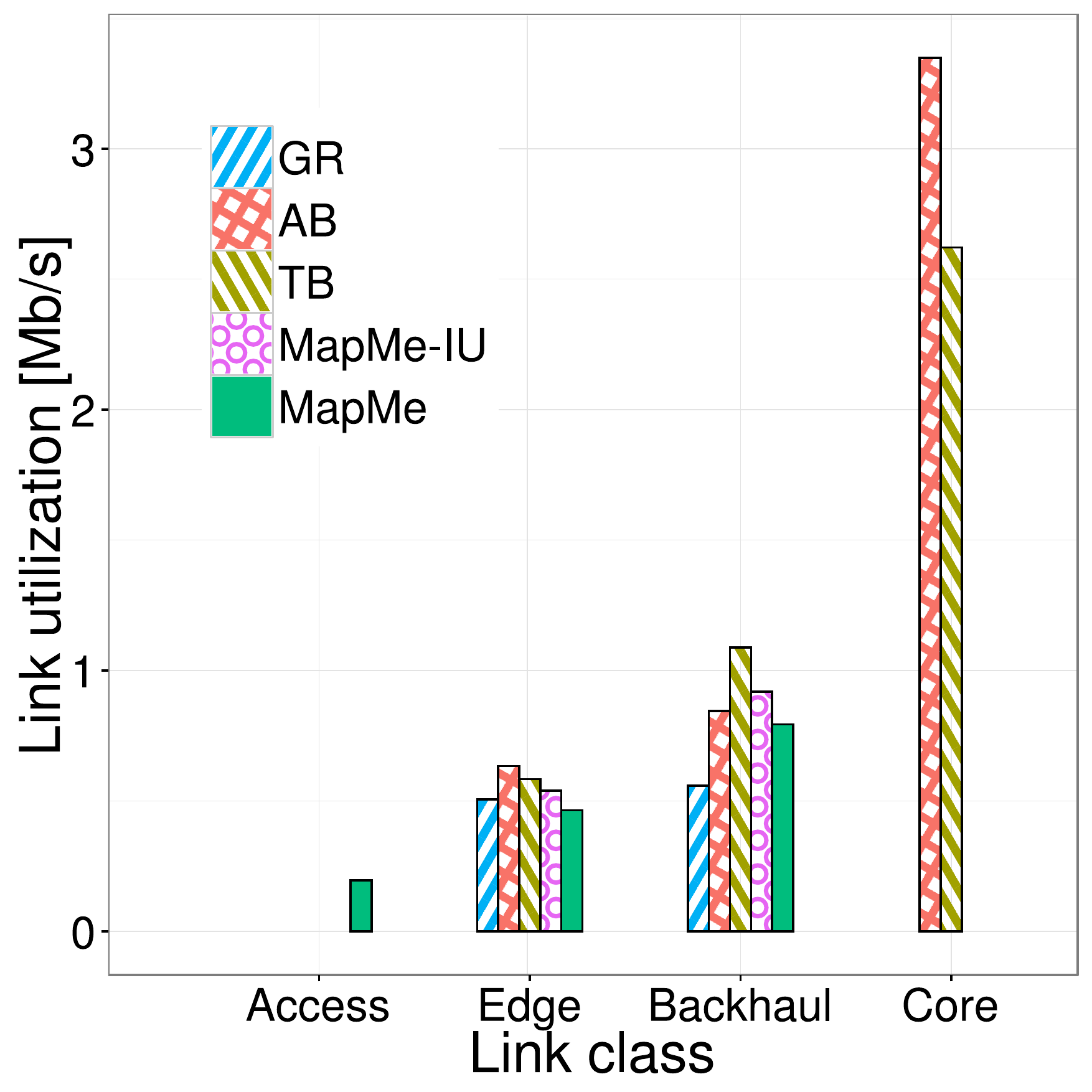}
		\label{fig:utilization}
	}
	\hfill
	\subfigure[] {
		\includegraphics[width=0.45\textwidth]{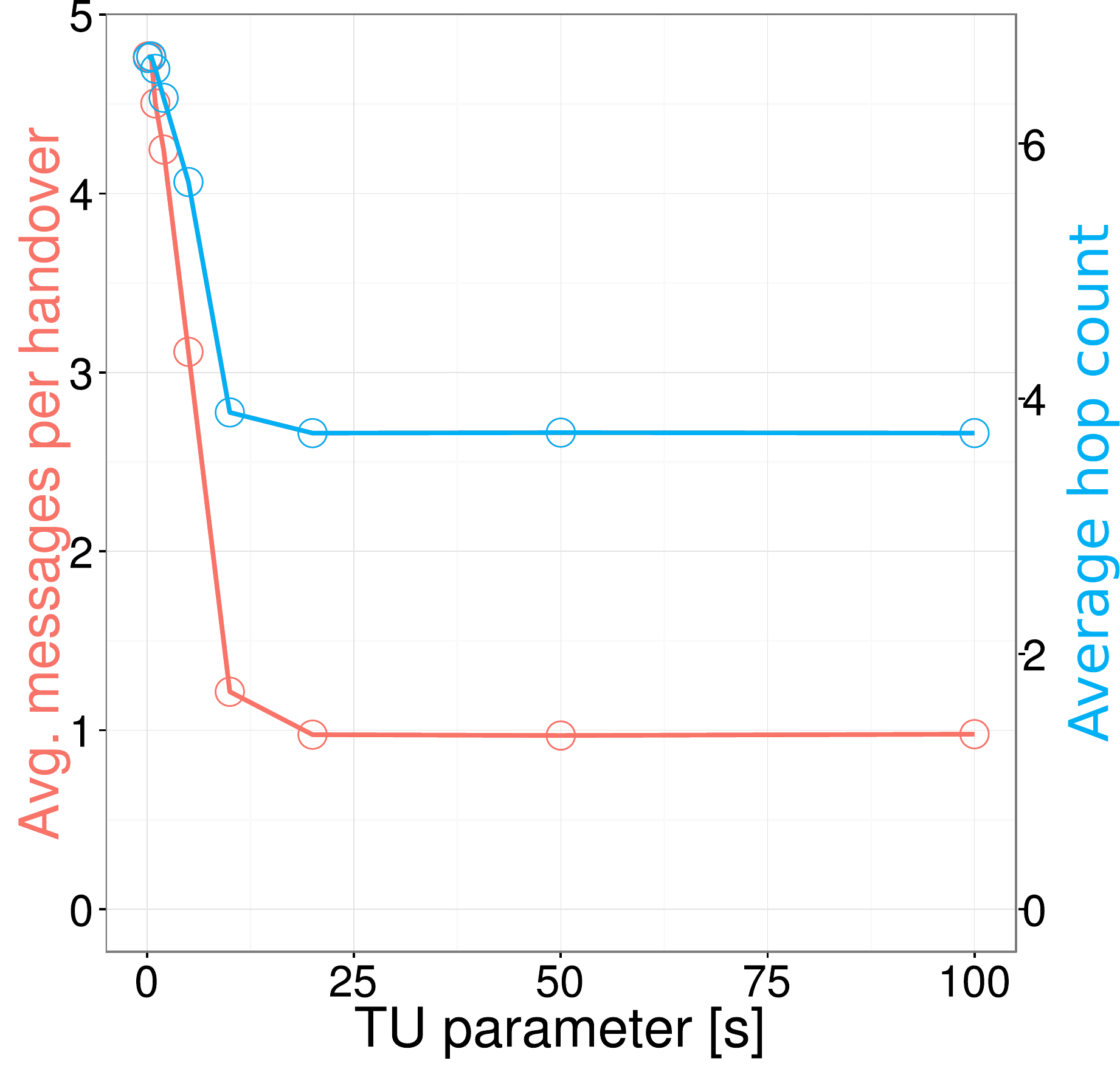}
		\label{fig:tu}
	}
	\vspace{-2mm}
	\caption{Network cost: Signaling overhead vs mobile speed (a); overhead (b) and link utilization (c) per router class. Map-Me sensitivity analysis (d).}
	\label{fig:performance-mapme}
	\vspace{-2mm}
\end{figure*}
For sake of completeness, tracing-based (TB) is obtained by enabling in Kite (\cite{zhang2014kite}) all optional optimizations
that we have found to be mandatory, in order to let the protocol work in presence of losses and frequent mobility:  
(i) traces are reissued right after an handover and (ii)  retransmitted if an ACK is not received withing 60ms, 
(iii) Interests from consumers make use of the trace table and the FIB also
(\texttt{traceonly} flag must be unset). In-network retransmission (\emph{pull})
is disabled for fair comparison.
 
\textbf{Topology and mobility}: We consider an IEEE 802.11n access network composed of 36 access points  
arranged in a 6 by 6 cells grid with square-shaped cells of side S=80 meters and connected to the fat-tree 
backhaul network in Fig.~\ref{fig:topology}.  Wired links have all capacity of C=10Mb/s and 5ms delay. 
Mobiles move only in a  4 by 4 cells sub-area to avoid border effects in the evaluations.
For AB/TB, we assume the anchor/home address to be placed at the root of the fat-tree, to achieve the best performance in terms of average path stretch as previously observed. To show both consumer/producer mobility, we assume $N$ disjoint pairs of consumer/producer, all connected to the leaves of the fat-tree. As mobility model, we use Random Way Point (RWP) in Sec.\ref{sec:synthetic} and trace-driven vehicular and pedestrian mobility patterns in Sec.\ref{sec:trace} (see \cite{Navigo}).

\textbf{ICN over IEEE 802.11n}: In this paper we use IEEE 802.11n on 5Gz frequencies, using a single base
channel of 40Mz with short guard intervals (SGI) using a single antenna at access point (AP) and at mobile stations (STAs). The PHY rate adaptation is minstrel (\cite{minstrel}) to set High Throughput (HT) rates, i.e. MCS from $0$ to  $7$
(corresponding to data rates from $15$Mbps to $150$Mbps). 802.11 frame aggregation is also enabled with a maximum A-MSDU size of $7935$ Bytes and A-MPDU maximum size of $64$kB with block acks.  ICN nodes interconnect to neighbors using adjacencies (denoted as faces), 
i.e. unicast bidirectional channels between two nodes established the same broadcast medium, 802.11n, which we assume in infrastructure mode. In such mode, faces are dynamically created and destroyed between STAs and APs following node associations.
When multiple STAs have a face established with the same AP, medium multiple access is managed by 802.11 EDCF which implies transmission latency and bandwidth sharing.  A face between a STA and the AP is characterized by a time varying capacity that
depends on a number of factors like radio conditions, rate selection, medium sharing. In this work we assume a small cell network wifi deployment controlled by a single entity that manages radio planning and 802.11 tuning. More precisely, we consider homogeneous grid cell coverage, with inter AP distance of 80 meters and  APs operating with a maximum power of 40mW (16dBm) for a maximum
radio range of 120 meters in outdoor.  This makes the maximum number of APs sensed by a STA to grow up to nine. 
Wifi STAs move from one cell to another in a fully covered geographic area and use the ysteresis handoff algorithm described in
\cite{Mhatre:2006} to perform handover. The Rice model is used to take into account fast fading with the line of sight.
See \cite{wifi} for more details about 802.11n parametrization.

\subsection{RWP mobility}\label{sec:synthetic}
In this section, we assume RWP mobility with uniform points distributed across a squared mobility area and constant speed. To highlight \name benefits in the support of latency-sensitve traffic, we consider a streaming audio/video application, characterized by a CBR rate of $1$Mbps. In the following we report statistics about user performance and network cost.

\subsubsection{User performance}
In Fig.\ref{fig:loss_rate}-\ref{fig:delay}, we show two performance indicators for latency-sensitive traffic,
average packet loss and delay, in case of $N=5$ consumer/producer pairs as a function of mobile speed (from pedestrian, i.e. $1$m/s or $3$km/h to vehicular, i.e. $15$m/s or $54$km/h).
We can distinguish two kind of losses: due to the wireless medium, occurring no matter the mobility management approach, and increasing with mobile speed, and those due to mobility. The fraction of mobility losses is consistently reduced by \name especially in presence of the notification/discovery mechanism, as a result of in-fly re-routing of Interests towards the new location of the producer, which prevents Interest timeouts. \nameu like TB (or alternative AL solutions) enable re-routing of
Interests only after the interval of time required for an update to complete. A longer time is required for a global routing update, but the resulting path is the shortest path possible, which explains the equivalent performance w.r.t. \nameunosp/TB. AB under-performs because of worse update completion time and path stretch. 
The experienced average packet delay in Fig.\ref{fig:delay} is a consequence of the path stretch of different approaches: high for AB, medium for TB or \nameunosp, low for GR as in Tab.\ref{tab:comparison} in Sec.\ref{sec:relwork}. \name achieves better performance especially at high speed when the discovery/notification mechanism is mostly used in virtue of the shorter $1$-hop forwarding between APs at the access that does not involve upper links in the topology (at the edge level).
As explained, packet losses and delay result from the different average path length associated to each mobility update process, see Fig.\ref{fig:stretch}, and from the L3 handoff latency, i.e. the time required for L3 reconnection after an handover, see Fig.\ref{fig:l3delay}. The L3 handoff latency illustrates the reactivity of the mobility management protocol and
highlights the significant improvement brought by \name which reduces latency to zero.
It is interesting to observe that AB shows a constant latency value of around $30$ms due to update propagation up to the anchor, while for GR, TB and \nameu such latency varies according to the number of routers to be updated, as a function of producer movement in the considered topology. Latency variations can be visualized at the inflection points in the corresponding CDFs in Fig.\ref{fig:l3delay}.

\begin{figure*}[htbp]
	\centering
	\subfigure[]{
		\includegraphics[width=0.45\textwidth]{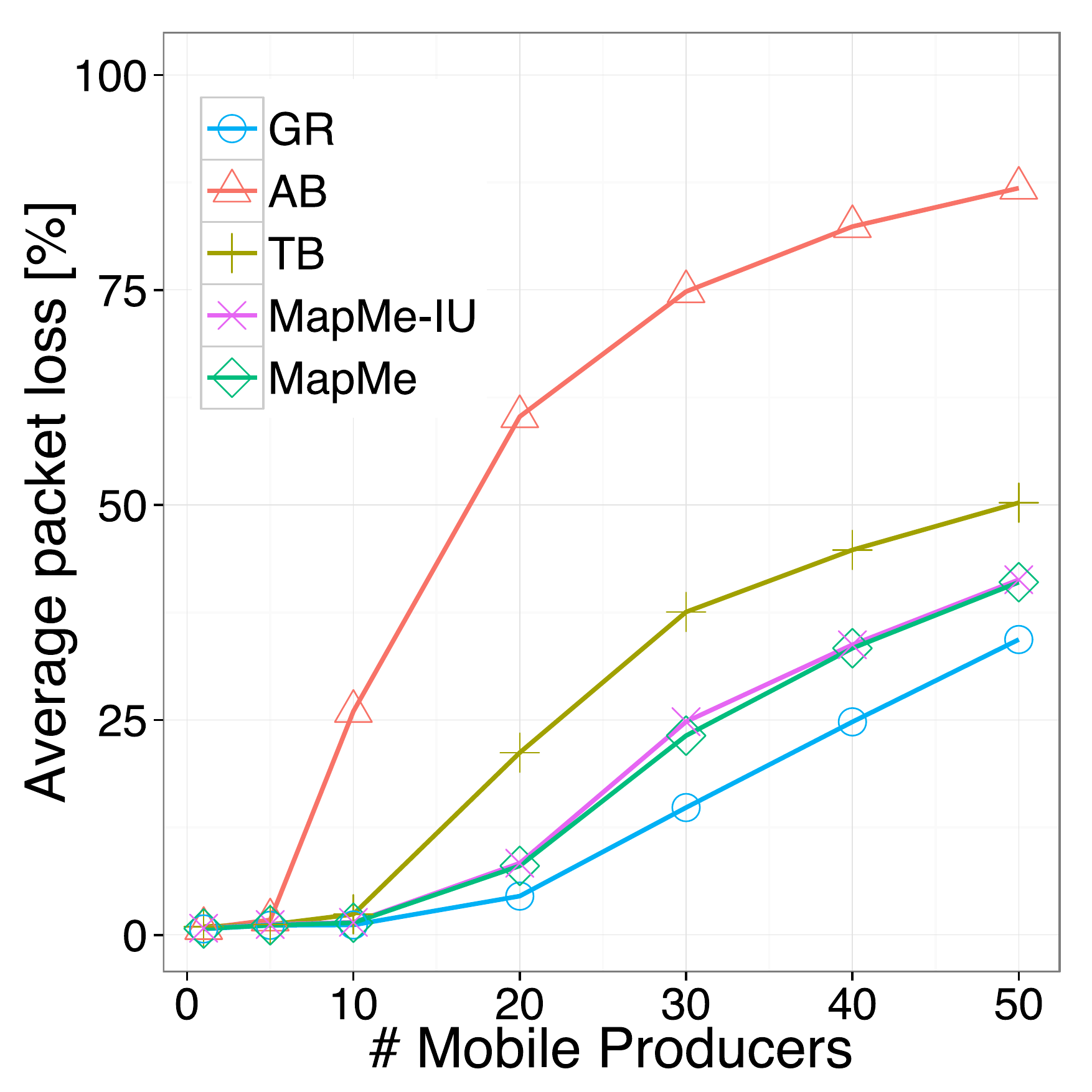}
		\label{fig:wild_pktloss_avg}
	}
	\hfill
	\subfigure[]{
		\includegraphics[width=0.45\textwidth]{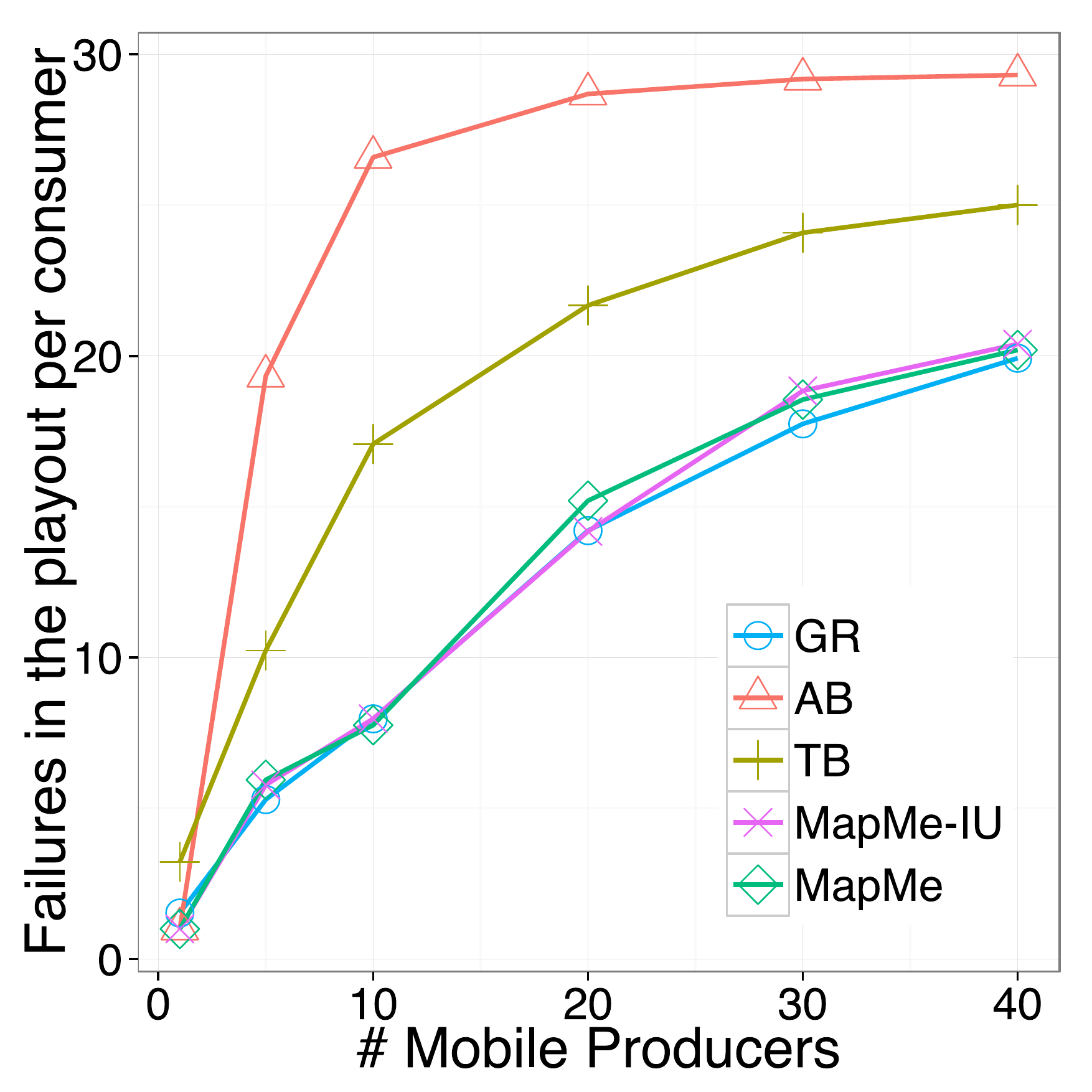}
		\label{fig:wild_periscopeFailures}
	}
	\subfigure[]{
		\includegraphics[width=0.45\textwidth]{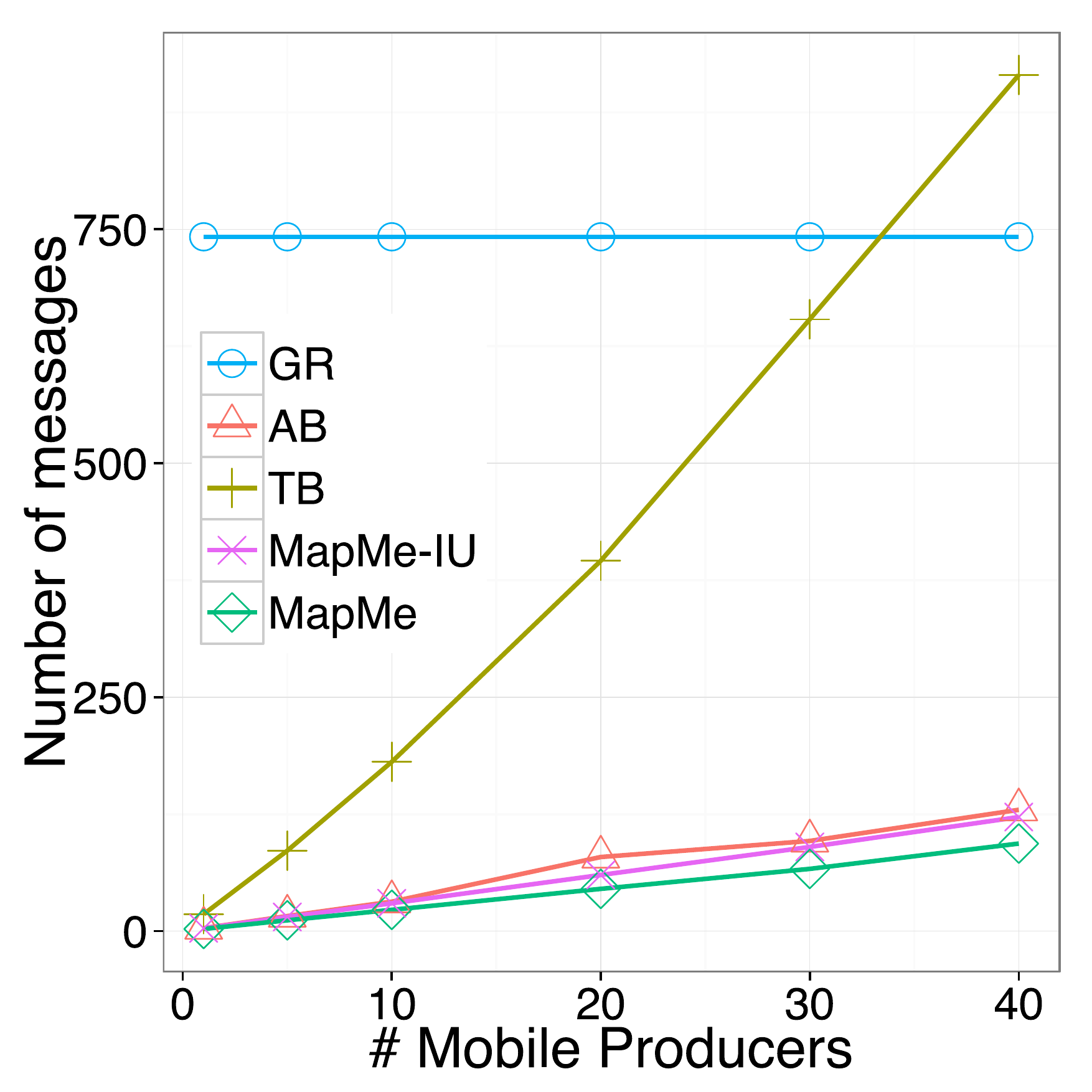}
		\label{fig:overhead-wild}	
	}
	\subfigure[]{
		\includegraphics[width=0.45\textwidth]{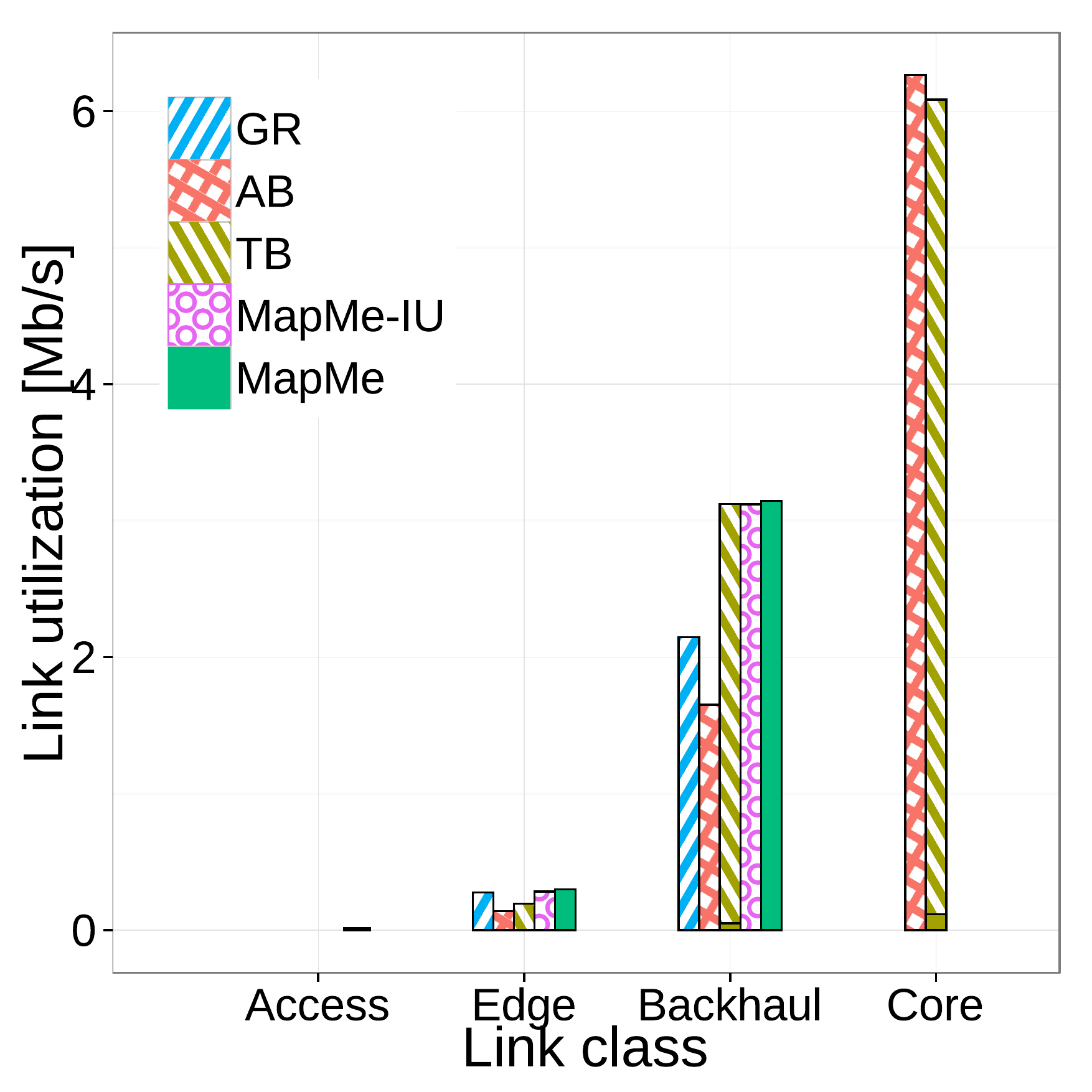}
		\label{fig:wild-link-load}	
	}
	\vspace{-2mm}
	\caption{User performance: CBR average packet loss (a), Periscope playout failures (b). Network Cost: CBR overhead (c) and Periscope link utilization (d).}
	\label{fig:CBRpktLoss-wild}
	\vspace{-3mm}
\end{figure*}
\subsubsection{Network cost}

If user performance is critical to drive mobility management choice,  network cost analysis is equally important for the selection of a cost-effective solution.  To this aim, we compare signaling overhead, meaning the total number of control messages triggered by an handover, in Fig.\ref{fig:overhead}, and the volume of signaling messages per handover to be processed by routers at different positions in the network, in Fig.\ref{fig:ratio}.
More precisely, in the latter case we visualize the distribution over the network of signaling load by distinguishing
the average number of messages per handover received by different classes of routers, based on their position in the network: access, edge, backhaul, core as indicated in  Fig.\ref{fig:topology}.
As expected, the overall number of signaling messages as a function of mobile speed is constant for AB, equal to the number of hops from mobile nodes to the anchor ($4$). Instead, it varies for \name and \nameu according to the also varying average hop count (i.e. path stretch), as already observed in Fig.\ref{fig:stretch}. TB approaches involve a much higher signaling overhead due to trace `keep-alive' messages periodically sent to refresh update information. 
By reporting the way traffic is spread across the network and where signaling traffic goes, we can draw some key observations.
Every mobility protocol relies on the control plane, that enforces a routing
state across the network (shortest-path routing in this paper), which
corresponds to the initialization state for mobility. All protocols relying on
an anchor have routing pointing to the anchor's location, whereas for AL
mechanisms, it points to the producer's position at routing update time. Thus,
AL approaches are able to offload mobile backhaul and core networks from all
local traffic, seamlessly (Fig.\ref{fig:utilization}\footnote{For clarity,
utilization of access link only represents traffic between base stations, excluding
upstream from mobiles.}). Finally, we report about \name sensitivity to
parametrization, i.e. the impact of TU settings. In Fig.\ref{fig:tu} we observe
that \name has robust parametrization as long as TU is not too small (signaling
overhead and path stretch quickly converge to best settings) or too high (load
on access).

\subsection{Trace-driven urban mobility}\label{sec:trace}
\textbf{Topology:} In order to evaluate our approach under more realistic mobility patterns, we consider an urban residential environment spanning a $2.1 \times 2.1$ km$^2$ area in Los Angeles, with the WiFi ot Spot deployment of Time Warner Cable~\cite{TWCwebsite}, i.e. we dislocate $729$ WiFi APs, with the same wireless settings as in previous experiments, connected to the Internet through the fat tree topology in Fig.\ref{fig:topology}.

\textbf{Mobility:} We generate realistic vehicular mobility patterns using SUMO~\cite{SUMO2012}, with maximum car speed set according to road speed limits\footnote{In the selected area we have three different road categories characterized by different speed limits: 40, 70 and 55 km/h.}. We place the mobile producer in a moving car and analyze system dynamics on a given time interval (4 minutes, roughly corresponding to 33 handovers), so that all monitored cars are in the map at the same time.
In such scenario we consider a group communication between one mobile producer and two non-mobile consumers requesting 
different data. Consumers are connected to two APs that are picked at random, uniformly across the network coverage.

\textbf{Applications:} Two type of applications are considered: in the first set of simulations, the previously detailed streaming application is characterized by a rate of 1Mbit/s. In the second set, a pseudo real-time video streaming application, reproducing the popular application Periscope~\cite{Periscope} is used.
The mobile producer generates two different video streams, each one downloaded and played by one consumer, using a 5 seconds play-out delay buffer. If the video play-out stops because the consumer has no Data available, we consider this as a failure and momentarily stop the consumer: after a short period of time (few seconds) the consumer restart downloading new data and to play-out the video. The video data rate is 1Mbit/sec, corresponding to a 480p video resolution. 
Traffic is scaled up by increasing the number of groups, identified by the producer serving data.

\subsubsection{User Performance}
To quantify user experience, we analyze the following metrics: the average packet loss and  user satisfaction, while varying the number of mobile producers in the area (from 1 to 50, each one serving two consumers).\\
\textbf{Packet loss:} We evaluate the distribution of packet losses per second for the CBR application. 
Fig.~\ref{fig:wild_pktloss_avg} shows the average packet loss for \name and the other protocols, while increasing the number of mobile producers in the system. As expected, increasing the number of active users in the network has a negative affect to the performance, because links are getting congested and routers start to lose packets. owever, as shown in Fig.~\ref{fig:wild_pktloss_avg}, the performance of \name and \nameu is close to the ideal GR, while 
TB leads to higher loss rate and with AB we observe a even more rapid increase in packet loss. Indeed, the distributed nature of \name allows the proposed solution to better cope with an increasing number of mobile producers.

\textbf{User Satisfaction:} We evaluate user satisfaction by analyzing the number of failures in the play-out of the video stream for the pseudo real-time video streaming (Periscope-like). Fig.~\ref{fig:wild_periscopeFailures} shows the number of failures in the video play-out that each consumer encounter in $4$ minutes. Like in the CBR case, when the number of mobile producers increases, the performance of the system degrades. Similarly to what observed in CBR case before, AB concentrates
all traffic on a single node, the anchor, thus giving rise to congestion. In contrast, distributed protocols such as \name are able to better distribute traffic over the network and thus better cope with larger number of users. For the same reason, TB performs better than AB, but worse than \namenosp/GR. Indeed, sending traces to the anchor force traffic towards upper layers in the network, preventing substantial traffic offload at the edge.
It is worth noticing that the application runs a classic window-based congestion control with no video rate adaptation and no specific mobility loss recovery mechanisms (all packet losses are recovered based on timer expiration at the consumer). The design of such schemes and the analysis of the interaction with mobility management protocols is out of the scope of this paper and left for future work.

\subsubsection{Network Cost}
Beyond user performance, we evaluate \name in terms of network cost, by computing the overhead and comparing it to 
all other considered solutions. Fig.\ref{fig:overhead-wild} reports the overhead
versus the number of producers, by computing the aggregate  number of messages
exchanged during handover for all producers, whereas
Fig.\ref{fig:wild-link-load} displays link load distribution across the network
(in the case of 10 mobile producers in the map). 
The figures prove that \name successfully offloads the core from local traffic with light overhead, in virtue 
of the anchor-less characteristics of its design. 

\section{Conclusions}\label{sec:conclusions}
Native support for mobility management at network layer is a recognized strength of ICN, and appears to be a key feature to exploit for the design of 5G networks. owever, a comprehensive solution for mobility management in ICN still lacks: previous attempts so far have either tried to apply Mobile IP concepts to ICN or looked at partial aspects of the problem, without providing a thorough evaluation of the initial solutions sketched in ICN context.
The contribution of this paper is twofold. First, we define \namenosp, an anchor-less model for managing intra-AS producer mobility even in presence of latency-sensitive traffic. By design, \name is simple as it only leverages ICN forwarding plane and reactive notifications to the network, lightweight in terms of required signaling messages and, to the best of our knowledge, the first with also proven guarantees of bounded stretch and overall correctness for the forwarding update process. 
Second, we develop a simulation framework on top of NDNSim 2.1 integrating \namenosp, anchor-based and tracing-based approaches and ideal global routing as a reference mobility model, using model-based and trace-driven consumer/producer mobility patterns.
Evaluation takes 802.11n access in small cell outdoor settings and proves how wifi can support mobility using ICN in general
settings.

Reported results show that \name optimally offloads the infrastructure from communications that are local.
All other approaches making use of an anchor, which in practice is also the network gateway, can be optimized
only if all traffic is non local. Instead, current propositions in 3GPP to offload the mobile network core stem from
the observation that, on the contrary, communications are most likely local.
On the other hand, \name would serve non local communications through one 
or multiple gateways without binding mobility feature to any specific location.



\small
\bibliographystyle{abbrv}
\bibliography{bibliography}
\end{document}

%% file: deps.inc
\usepackage{amsthm}
\usepackage{float}

\usepackage{url}
\usepackage{graphicx}
\usepackage[lined, ruled]{algorithm2e}
\usepackage{algpseudocode}
\usepackage{epstopdf}
\usepackage{color}
\usepackage{booktabs}
\usepackage{subfigure}
\usepackage{epsfig}
\usepackage{comment}
\usepackage{cases}
\usepackage{float}
\usepackage{tabularx}
\usepackage[T1]{fontenc}
\usepackage[utf8]{inputenc}
\usepackage{upgreek}
\usepackage{marginnote}
\usepackage{fullpage}
\usepackage{scrextend}

\makeatletter
\let\oldmarginnote\marginnote
\renewcommand*{\marginnote}[1]{%
   \begingroup%
   \ifodd\value{page}
     \if@firstcolumn\reversemarginpar\fi
   \else
     \if@firstcolumn\else\hspace{3cm}\reversemarginpar\fi
   \fi
   \oldmarginnote{#1}%
   \endgroup%
}
\makeatother

\ifpdf
    \DeclareGraphicsExtensions{.jpg,.pdf,.mps,.png}
\else
    \usepackage{hyperref}
    \DeclareGraphicsExtensions{.eps,.ps,.bmp,.tif,.tiff,.tga}
\fi
\graphicspath{{fig/}}

\usepackage{xcolor}

\newsavebox{\mybox}
\newlength{\mydepth}
\newlength{\myheight}
{\begin{lrbox}{\mybox}\begin{minipage}{\columnwidth}}%
{\end{minipage}\end{lrbox}%
 \settodepth{\mydepth}{\usebox{\mybox}}%
 \settoheight{\myheight}{\usebox{\mybox}}%
 \addtolength{\myheight}{\mydepth}%
 \noindent\makebox[0pt]{\hspace{-20pt}\rule[-\mydepth]{3pt}{\myheight}}%
 \usebox{\mybox}}

\newcommand{\Let}[2]{#1 $\leftarrow$ #2}

\newtheorem{proposition}{Proposition}
\newtheorem{corollary}{Corollary}

\renewcommand{\Comment}[1]{%
  {\textcolor{gray}{$\triangleright$\emph{#1}}}
}